\documentclass[11pt,letterpaper]{article}

\usepackage{amsmath,amsthm,amsfonts,amssymb}
\usepackage{thm-restate}
\usepackage{fullpage}
\usepackage[utf8]{inputenc}
\usepackage[dvipsnames]{xcolor}
\usepackage{xspace,enumerate}
\usepackage[shortlabels]{enumitem}
\usepackage[hypertexnames=false,colorlinks=true,urlcolor=Blue,citecolor=Green,linkcolor=BrickRed,breaklinks=true]{hyperref}
\usepackage[capitalise]{cleveref}
\usepackage[OT4]{fontenc}
\usepackage[ruled]{algorithm}
\usepackage[noend]{algorithmic}
\usepackage{ifpdf}
\usepackage{enumerate}
\usepackage{thmtools}
\usepackage{url}

\usepackage[margin=1in]{geometry}
\title{Deterministic Fully Dynamic SSSP and More}

\newcommand{\email}[1]{\href{mailto:#1}{#1}}

\date{\vspace{-5ex}}
\author{Jan van den Brand\thanks{School of Computer Science, Georgia Institute of Technology, USA. \email{vdbrand@gatech.edu}.}
\and Adam Karczmarz\thanks{University of Warsaw and IDEAS NCBR, Poland. \email{a.karczmarz@mimuw.edu.pl}. Partially supported by the ERC CoG grant TUgbOAT no 772346 and the National Science Centre (NCN) grant no. 2022/47/D/ST6/02184.}}

\theoremstyle{plain}
\newtheorem{problem}{Open Problem}
\crefname{problem}{open problem}{open problems}
\newtheorem{theorem}{Theorem}[section]
\newtheorem{lemma}[theorem]{Lemma}
\newtheorem{corollary}[theorem]{Corollary}

\newtheorem{remark}[theorem]{Remark}

\def\poly{\operatorname{poly}}
\def\polylog{\operatorname{polylog}}

\newcommand{\mA}{\mathbf{A}}
\newcommand{\mB}{\mathbf{B}}
\newcommand{\mL}{\mathbf{L}}
\newcommand{\mR}{\mathbf{R}}
\newcommand{\mM}{\mathbf{M}}
\newcommand{\mN}{\mathbf{N}}
\newcommand{\mI}{\mathbf{I}}
\newcommand{\mU}{\mathbf{U}}
\newcommand{\mV}{\mathbf{V}}
\newcommand{\F}{\mathbb{F}}
\newcommand{\Z}{\mathbb{Z}}

\newcommand{\Ot}{\ensuremath{\widetilde{O}}}
\newcommand{\eps}{\ensuremath{\epsilon}}
\newcommand{\dist}{\delta}

\newcommand{\wei}{w}

\begin{document}

\maketitle
\pagenumbering{roman}
\begin{abstract}
We present the first non-trivial fully dynamic algorithm maintaining exact single-source distances in unweighted graphs.
This resolves an open problem stated in~\cite{Sankowski05,BrandN19}.
Previous fully dynamic single-source distances data structures were all approximate, but so far, non-trivial dynamic algorithms for the exact setting could only be ruled out for polynomially weighted graphs~\cite{AbboudW14}. 
The exact unweighted case remained the main case
for which neither a subquadratic dynamic algorithm nor a quadratic lower bound was known.

Our dynamic algorithm works on directed graphs and is deterministic, and can report a single-source shortest paths tree
in subquadratic time as well.
Thus we also obtain the first deterministic fully dynamic data structure for reachability (transitive closure) with subquadratic update and query time. This answers an open problem from \cite{BrandNS19}.

Finally, using the same framework we obtain the first fully dynamic data structure maintaining
all-pairs $(1+\eps)$-approximate distances within non-trivial sub-$n^\omega$ worst-case update
time while supporting optimal-time approximate shortest path reporting at the same time. 
This data structure is also deterministic and therefore implies the first known non-trivial \emph{deterministic worst-case} bound for recomputing the transitive closure of a digraph.

\end{abstract}
\newpage
\tableofcontents
\newpage
\pagenumbering{arabic}

\section{Introduction}

The dynamic shortest path problem is a classic problem in graph algorithms that seeks to maintain the shortest path distances between pairs of vertices in a graph as the graph is subject to updates such as edge insertions, deletions, or weight changes.
This problem and its variants have been extensively studied in the theoretical computer science literature, and different algorithmic approaches have been proposed and analyzed depending on the nature of the updates (incremental \cite{AusielloIMN91,GutenbergWW20,HenzingerKN-ICALP13,Karczmarz0S22,KyngMG22,ChechikZ21}, decremental \cite{BaswanaHS07,EvaldFGW21,EvenS81,BernsteinC16,BernsteinGS20,BernsteinGW20,BernsteinR11,Bernstein16,ChuzhoyS21} or fully dynamic \cite{King99,DemetrescuI01,Thorup04,Thorup05,HenzingerK95,BrandFN22,BergamaschiHGWW21,Sankowski05,AlokhinaB23,AbrahamCG12}) and the type of queries ($st$ \cite{BrandNS19,HenzingerKNS15,AbboudW14}, single-source \cite{BrandN19,BernsteinC17,Bernstein17,ChuzhoyK19,Chuzhoy21,HenzingerKN14,GutenbergW20a,BernsteinGS21} or all-pairs \cite{DemetrescuI04,AbrahamCK17,GutenbergW20b,ChechikZ23,Mao23,DemetrescuI02,HenzingerKN13,Bernstein09,RodittyZ12,AbrahamCT14}).
Here, ``incremental'' refers to the setting with only edge insertions, ``decremental'' to only edge deletions, and ``fully dynamic'' to the general setting where a graph undergoes both insertions and deletions.

Among these different variants, single-source is one of the most well-studied.
It offers a good trade-off between generality and speed, as it maintains several distances at once while allowing much faster update times than the all-pairs variant. 
It also has a simple trivial solution by recomputing the shortest paths tree in $O(m)$ time ($O(n^2)$ time for dense graphs) via Dijkstra's algorithm after each update, which serves as a natural benchmark for all improvements.
As a result, dynamic single-source shortest paths is one of the most prolific problems in the dynamic graphs area with major results for different variants such as
incremental~\cite{GutenbergWW20,ChechikZ21,HenzingerKN-ICALP13,KyngMG22}, decremental~\cite{HenzingerKN14,GutenbergW20a,EvenS81,BernsteinC16,BernsteinC17,Bernstein17,BernsteinGS20,BernsteinGW20,BernsteinGS21,BernsteinR11,ChuzhoyK19},
and fully dynamic~\cite{RodittyZ11,BrandN19,BrandFN22,BergamaschiHGWW21,AlokhinaB23}.
Yet, despite being one of the most studied variants of dynamic shortest paths, the case of \emph{exact} fully dynamic single source shortest paths -- even in dense graphs with $\Theta(n^2)$ edges -- has still remained open and withstood all attempts to construct efficient data structures.
Even in the weakest graph model, i.e.~undirected and unweighted graphs,
no exact fully dynamic dynamic algorithm is known other than trivial recomputation of the shortest paths after each update via Dijkstra's algorithm.
Previously, progress has only been made for the $st$-case \cite{Sankowski05,BrandNS19,BrandFN22,KarczmarzS23} where an exact fully dynamic distance can be maintained on unweighted graphs in $O(n^{1.673})$ time (i.e., subquadratic in $n$ and sublinear in the number of edges for a sufficiently large density), but that approach does not extend to the single-source case.
It has been asked repeatedly to find a dynamic algorithm for dynamic single source shortest paths since almost two decades ago; see, e.g.,~\cite{DemetrescuI04,Thorup04}.
One reason why an efficient data structure has been elusive is that on weighted graphs (with sufficiently large positive integer weights allowed), it is impossible to obtain a dynamic algorithm with subquadratic update time unless the APSP conjecture is false \cite{RodittyZ11,AbboudW14}. 
So a natural research direction is to find an exact algorithm for the unweighted case (and any density) which was also explicitly stated 
as an open problem in \cite{Sankowski04,BrandN19}.
\begin{problem}[\cite{DemetrescuI04,Thorup04,Sankowski04,BrandN19}]\label{problem:sssp}
  Can one maintain fully dynamic exact shortest paths in subquadratic time (in the number of vertices $n$) on unweighted graphs?
\end{problem}

It was shown by Abboud and Vassilevska Williams \cite{AbboudW14} that dynamic $st$-reachability and $st$-shortest path require algebraic techniques (i.e.,~fast matrix multiplication) if one hopes to beat the naive $O(n^2)$ update time algorithm of running BFS/Dijkstra after each update.
So, an affirmative answer to \cref{problem:sssp} is only possible if we rely on algebraic techniques.
In previous work, the algebraic perspective requires dynamic algorithms to perform calculations involving very large numbers,
 and to reduce the bit-length of the numbers involved, past work relied on the randomized Schwartz-Zippel lemma (i.e.,~randomized polynomial identity testing) \cite{Sankowski04,BrandNS19,BrandS19} or performing all arithmetic modulo a random prime \cite{KingS02,DemetrescuI00}.
As a result, the vast majority of algebraic data structures for dynamic graph problems is randomized \cite{Sankowski04,Sankowski05,Sankowski07,SankowskiM10,BergamaschiHGWW21,BrandNS19,BrandN19,BrandS19,GuR21}.

Deterministic algorithms would be preferable because they allow for greater confidence in their correctness and efficiency, as opposed to randomized algorithms which may only offer probabilistic guarantees on correctness or running time that do not hold in all cases. Moreover, randomized algorithms may restrict the ``adversary'' creating the sequence of updates from being adaptive. In contrast, deterministic algorithms avoid these issues and allow for more predictable behavior.

For algebraic data structures, only in the case of shortest paths in the incremental setting \cite{Karczmarz0S22} or on undirected graphs in the approximate setting \cite{BrandFN22} has derandomization been possible, which naturally excludes fully dynamic reachability or fully dynamic shortest paths on directed graphs.
Finding a way to derandomize fully dynamic reachability has been asked in \cite{BrandNS19}.
\begin{problem}[{\cite{BrandNS19}}]\label{problem:deterministic}
	Can fully dynamic $st$-reachability be solved \emph{deterministically} in subquadratic time?
\end{problem}
In other dynamic graph settings, derandomization has received a lot of attention \cite{ChechikZ23,BrandFN22,Lacki11,BernsteinGS21,BernsteinGS20,Bernstein17,BernsteinC17,BernsteinC16,Chuzhoy21,AbrahamCK17,ChuzhoyS21,HenzingerKN13,Karczmarz0S22}, but for fully dynamic reachability with subquadratic update time the question has remained open.

In this work, we answer both \cref{problem:sssp,problem:deterministic}. We construct a fully dynamic algorithm that maintains \emph{exact} single-source shortest paths in subquadratic update time, thus providing an answer to \cref{problem:sssp}.
This dynamic algorithm is deterministic and works on directed graphs, and can thus solve fully dynamic $st$-reachability deterministically, solving \cref{problem:deterministic}.

Finally, the framework that we use is powerful and flexible enough to also yield
new worst-case update bounds for fully dynamic $(1+\eps)$-approximate all-pairs shortest paths.

\subsection{Our results}

\paragraph{SSSP} We present the first fully dynamic algorithm that can maintain \emph{exact} single source distances, answering \cref{problem:sssp}. 
\begin{theorem}\label{thm:intro:singlesource}
	Let $G=(V,E)$ be unweighted digraph. There exist the following data structures supporting
	fully dynamic single-edge updates to $G$ and answering queries about exact distances from an arbitrary query source vertex $s\in V$ to all other vertices in $G$.
	\begin{itemize}
    \item A deterministic data structure with $O(n^{1.969})$ worst-case update and query time. (Cf.~\Cref{t:ss-distances}) 
		\item A Monte Carlo randomized data structure with $O(n^{1.933})$ worst-case update and query time. (Cf.~\Cref{t:ss-distances-randomized})
	\end{itemize}
\end{theorem}

The algorithm works on unweighted directed graphs but could also be easily generalized to handle small (e.g., sufficiently low-polynomial) integer weights $[1,W]$ at the cost of a $\poly(W)$ multiplicative time overhead.\footnote{Just like other algebraic data structures for exact distances~\cite{BrandFN22, Sankowski05, KarczmarzS23}. See e.g.~\cite{KarczmarzS23}.} 
Recall that an exact algorithm with subquadratic update time for graphs with integer weights in $[1,n^c]$, where $c>0$ is a sufficiently large constant, is impossible under the APSP conjecture \cite{AbboudW14}. 

Previous results could maintain single-source shortest paths (or single-source distances) only in the approximate setting \cite{BrandN19,BergamaschiHGWW21,BrandFN22}, or in the partially dynamic setting \cite{HenzingerKN14,GutenbergWW20,GutenbergW20a,HenzingerK95,HenzingerKN-ICALP13,EvenS81,BernsteinC16,BernsteinC17,Bernstein17,BernsteinGS20,BernsteinGW20,BernsteinGS21,BernsteinR11}.

The data structures behind~\Cref{thm:intro:singlesource} are algebraic, i.e.,~make use of fast matrix multiplication. Algebraic techniques are inevitable to obtain a non-trivial update time \cite{AbboudW14}.
That being said, using any subcubic matrix multiplication algorithm (e.g., Strassen's) would yield subquadratic bounds in~\Cref{thm:intro:singlesource}.
If $\omega=2$ (here $n^\omega$ is the number of operations required to multiply two $n\times n$ matrices; currently it is known that $\omega\leq 2.372$~\cite{fmm,WilliamsXXZ23}), the deterministic update bound becomes $\Ot(n^{2-1/30})$ and the randomized one becomes $\Ot(n^{2-1/14})$.

Both variants can be extended, deterministically in a black-box way, to support \emph{SSSP tree} reporting queries within slightly larger, albeit still subquadaratic time (the latter one against an adaptive adversary). Concretely, if $\omega=2$, the tree-reporting bounds are $\Ot(n^{2-1/60})$ and $\Ot(n^{2-1/28})$, respectively, while for current bounds on $\omega$ \cite{AlmanW21,GallU18,fmm} they are $O(n^{1.985})$ and $O(n^{1.967})$, respectively.
Note our dynamic algorithms work on directed graphs, so they can be also used to solve fully dynamic $st$-reachability deterministically. This answers \cref{problem:deterministic}, asked in \cite{BrandNS19}. 

\paragraph{All-pairs problems} 
By reusing and extending the framework used for fully dynamic single-source data structures, we also obtain a new data structure for fully dynamic approximate all-pairs shortest paths, captured by the following \Cref{t:tc-reporting}.
\begin{restatable}{theorem}{tcreporting}\label{t:tc-reporting}
	Let $G=(V,E)$ be an unweighted directed graph. Let $\eps\in (0,1)$. There exists a deterministic data structure explicitly maintaining 
	all-pairs $(1+\eps)$-approximate distances in $G$ under vertex updates in $O(n^{2.293}/\eps)$ worst-case time per update.
	
	The data structure can be extended to:
	\begin{itemize}
		\item enable reporting a $(1+\eps)$-approximately shortest path $P=s\to t$ in $G$
		for query vertices $s,t\in V$ in optimal $O(|P|)$ time.
		\item explicitly maintain \emph{reachability trees} from
		all the sources $s\in V$ within the same update bound.
	\end{itemize}
\end{restatable}

As is typically the case for fully dynamic data structures maintaining all-pairs information
explicitly (e.g., \cite{Sankowski04,DemetrescuI04,Thorup04,AbrahamCK17,GutenbergW20b,ChechikZ23,Mao23}), the data structure supports more general \emph{vertex updates} that
can insert/delete any number of edges incident to a single vertex.

To the best of our knowledge, ours is the first fully dynamic APSP data structure for \emph{directed} graphs
allowing optimal path reporting whose worst-case update time breaks through
the $\Ot(n^\omega/\eps)\approx \Ot(n^{2.372})$ static $(1+\eps)$-approximate APSP bound of~\cite{Zwick02}.
In fact, this is the case regardless of whether randomization is used.
Previous results breaking the $\Ot(n^\omega)$-barrier~\cite{BrandN19} could not report the path.
Optimal path reporting is supported by the purely-combinatorial data structures~\cite{DemetrescuI04,Thorup04,AbrahamCK17,GutenbergW20b,ChechikZ23,Mao23} but
their worst-case update time is at least $\Omega(n^{2.5})$.
Better worst-case update bounds for path-reporting APSP have been obtained using algebraic techniques~\cite{BergamaschiHGWW21,AlokhinaB23} but with \emph{superlinear} (in $n$, let alone in the distance between the queried vertices) query time
and necessarily under less general single-edge updates.\footnote{This is also achieved by our SSSP data structures.}
Thus, prior to our work, no non-trivial worst-case bound for fully dynamic APSP (in directed graphs) with optimal-time path reporting has been known and recompute-from-scratch~\cite{Zwick02} could be considered the state-of-the-art solution (indeed,
the static $(1+\eps)$-approximate APSP algorithm of~\cite{Zwick02} can efficiently reconstruct paths).

The data structure of Theorem~\ref{t:tc-reporting} is deterministic. Consequently, it also constitutes the first \emph{deterministic} APSP data structure for \emph{directed graphs} with sub-$n^\omega$ worst-case update time maintaining all-pairs distance estimates \emph{explicitly}. The fully dynamic APSP data structure of~\cite{DemetrescuI04, Thorup04} is deterministic and has near-optimal $\Ot(n^2)$ amortized update bound, but the worst-case update bound is cubic. 
Using randomization, a matrix of $(1+\eps)$-approximate distances can be maintained
in $O(n^{2.045}/\eps^2)$ worst-case time per update~\cite{BrandN19}.
As deterministic variants of~\cite{AbrahamCK17, GutenbergW20b} have $\Omega(n^{2.5})$~update bound, again, recompute-from-scratch~\cite{Zwick02} has been the state-of-the-art in the deterministic regime.

The above remains true even for the easier problem of explicitly maintaining the transitive closure under fully dynamic updates.
Whereas an optimal $O(n^2)$ randomized worst-case update bound has been known~\cite{Sankowski04}, no \emph{deterministic} data structure with \emph{worst-case} update time beyond the $O(n^\omega)$-time recompute-from-scratch~\cite{Munro71} has been found. Only \emph{amortized} deterministic $O(n^2)$ update bounds have been known~\cite{DemetrescuI08,King99,Roditty08}.
Theorem~\ref{t:tc-reporting} implies that the transitive closure can be recomputed deterministically in $O(n^{2.293})$ time (worst-case) after a vertex update.
Additionally, within that bound, all-pairs reachability trees -- providing an explicit representation of paths between all pairs of vertices -- can also be constructed.

Perhaps surprisingly, for a sufficiently large improvement to $\omega$ (concretely, if $\omega\leq 2.12$), the update time of the data structure of Theorem~\ref{t:tc-reporting} (depending on the value of $\omega$) may become polynomially smaller than $O(n^\omega)$, i.e., slower than recomputing from scratch.
It is interesting to ask whether it is possible to reproduce any of the accomplishments of that data structure \emph{unconditionally}, that is, within sub-$n^\omega$ time as long as $\omega>2$.\footnote{Note that if $\omega=2$, then all-pairs reachability can be recomputed from scratch in $O(n^2)$ time~\cite{Munro71} which is already the output size, so in that case no non-trivial dynamic algorithm can exist.}
Using an entirely different approach, we show this is indeed possible
for deterministic fully dynamic transitive closure.

\begin{restatable}{theorem}{tcdeterministic}\label{t:tc-deterministic}
	Let $G=(V,E)$.
	Unless $\omega=2$, there exists a constant $\eps_\omega$, $0<\eps_\omega<\omega-2$, dependent on $\omega$,
	such that
	one can deterministically and explicitly maintain the transitive closure of $G$ under fully dynamic vertex updates
	in $O(n^{\omega-\eps_\omega})$ worst-case time per update.
\end{restatable}
The drawback of the above data structure is that it cannot report paths in optimal time (in their length) but
merely in $O(n)$ time per query.

\subsection{Further related work}
All-pairs shortest paths (in directed graphs) and transitive closure have also been studied extensively in partially dynamic
settings~\cite{BaswanaHS07, Bernstein16, EvaldFGW21, FrigioniMNZ01, Italiano86, KarczmarzL19, KarczmarzL20,DoryFNV22}.
Other variants related to distances are dynamic diameter and extreme distances \cite{AnconaHRWW19,BrandN19}.
There also exists work with much larger approximation ratio, e.g.~\cite{ForsterGNS23,AbrahamCT14,ForsterGH21}.
In addition to upper bounds, dynamic distance problems have also been studied from a (conditional) lower bound perspective, see e.g.~\cite{AbboudW14,HenzingerKNS15,AbboudBF22,AbboudBKZ22,AnconaHRWW19,BrandNS19,GutenbergWW20,JinX22,RodittyZ11}.

\subsection{Technical ideas}

Our results are obtained via a combination of (inevitable) algebraic methods and a wide range of purely combinatorial techniques. The fully dynamic single-source shortest distances data structures follow the high-level approach of the data structures for obtaining worst-case update bounds for exact fully dynamic \emph{all-pairs shortest paths}~\cite{AbrahamCK17, GutenbergW20b, Thorup05, ChechikZ23, Mao23}.
These purely combinatorial data structures are, however, super-quadratic in many different aspects and in fact seem to not be able to break through the $\Omega(n^{2.5})$ barrier for the problem they solve.

We observe that a majority of computational bottlenecks therein can be delegated to the algebraic data structures based on dynamic matrix inverse~\cite{BrandNS19, Sankowski04} -- again, with a few purely combinatorial extensions applied on top, such as the recent path-reporting techniques~\cite{Karczmarz0S22}.
The power of algebraic data structures is also exploited to a greater extent than before.
Previous work relied on maintaining either the entire matrix inverse (e.g.~\cite{Sankowski04}), querying single cells or rows (\cite{BrandN19, BrandNS19, Sankowski04}) or 
maintaining well-structured (e.g., square) submatrices~\cite{Sankowski05, BrandFN22}.
On the other hand, we crucially need maintaining \emph{unstructured and dynamic} subsets of cells of the inverse of subquadratic size and adapt the algebraic tools to this scenario.

Finally, a common theme of all our deterministic algorithms
is using weaker notions of hitting sets of ``long'' shortest paths between pairs of vertices whose distance exceeds some threshold~$h$. In previous dynamic shortest paths data structures, $\Ot(n/h)$-sized hitting sets are typically randomly sampled (once) or deterministically recomputed (after an update) to hit paths in the \emph{current graph}. 

Our single-source distances data structure, on the other hand, only employs a (deterministically computed) hitting set to sparsify an explicit \emph{static} collection of $\Theta(n^2)$ paths
that is reinitialized not very often. As a result, we completely avoid the bottleneck of all previous dynamic algebraic distance data structures for directed graphs: maintaining the \emph{current} exact $h$-bounded distances between all pairs of vertices from the hitting set.
For example, it is not clear how this component can be derandomized while keeping the update time subquadratic.

In the data structure of Theorem~\ref{t:tc-deterministic}, a different relaxation of hitting sets (specific to reachability, where hitting \emph{shortest paths} is not required for vertices far apart) is exploited. In particular, we describe a deterministic trade-off between the size of such a relaxed hitting set and the time needed to compute it, which we believe may be of independent interest.

\subsection{Organization}
We start by giving preliminaries in \Cref{sec:prelim} and then give a more detailed technical outline of our algorithms in \Cref{sec:outline}.
Our algorithms rely on constructing hitting sets that require low-hop shortest paths.
We describe in \Cref{s:blocks} how to extract these required paths from a dynamic bounded-distances algorithm. 
Next, in \Cref{s:collection} we describe a subcubic preprocessing step that, using the algebraic path-reporting data structures, finds a collection of paths that is robust (in a sense) in the decremental setting.
Finally, we employ that preprocessing of \Cref{s:sssp} to obtain our fully dynamic SSSP data structure (\Cref{thm:intro:singlesource}). 
The extension to fully-dynamic path-reporting APSP (\Cref{t:tc-reporting}) is described in \Cref{s:tc-reporting}. The deterministic fully dynamic transitive closure data structure beating the $n^\omega$-barrier for all $\omega>2$ (\Cref{t:tc-deterministic}) is given in \Cref{s:tc-det}.

\section{Preliminaries}
\label{sec:prelim}
Let $G=(V,E)$ be a directed graph. For $D\subseteq V$, we denote by $G-D$ the graph $G$ with all edges incident
to the vertices $D$ removed.

For any $s,t\in V$ and $k\geq 0$,
$\dist_G^k(s,t)$ is the minimum length of an $s\to t$ path in $G$ with at most $k$ edges (hops).
If no such path exists, $\dist_G^k(s,t)=\infty$.
A~\emph{shortest $k$-bounded} $s\to t$ path in $G$ is an $s\to t$ path with at most $k$ hops and length $\dist_G^k(s,t)$.
The $s,t$-distance $\dist_G(s,t)$ is defined as $\dist_G^n(s,t)$.

\section{Technical outline}
\label{sec:outline}
This section focuses, in the first place, on giving a technical overview of the deterministic
fully dynamic SSSP data structure. We then briefly describe the changes needed for
path-reporting $(1+\eps)$-approximate APSP, and the possible improvements using randomization.

\subsection{Algebraic toolbox and previous approaches}
\label{sec:outline:randomizedtools}
Recall that by the conditional lower bound of~\cite{AbboudW14}, using algebraic fast matrix multiplication-based techniques is inevitable even for fully dynamic $st$-reachability.
In the following we review how the state-of-the art randomized data structures for $st$-distances and reachability use algebraic techniques, explain the limitations of their approaches, and introduce the algebraic data structure that we use.

Let $h\in [1,n]$ be a threshold parameter setting the boundary between ``short'' shortest paths with at most $h$
edges, and ``long'' shortest paths with more than $h$ edges.
Just like all previously described algebraic data structures maintaining shortest paths, we leverage dynamic matrix inverse data structures
for maintaining the $h$-bounded distances between certain pairs of vertices.

\paragraph{Path counting} 
The main idea behind algebraic $h$-bounded distances data structures is counting (possibly implicitly) paths
of given lengths $0,\ldots,h$ for all source-target pairs. Note that the number of length-$k$ paths between
a given pair of vertices of an unweighted graph is no more than $n^k$,
i.e., could take $\Ot(k)$ bits to be represented.
As a result, storing the required path counts for a given pair might require
up to $\Ot(h^2)$ bits.
It is the most convenient to represent the counts as a degree-$h$-polynomial matrix, where
the coefficient of $x^k$ of the polynomial in the cell $(s,t)$ equals the count of $s\to t$ paths of length $k$.
As originally observed by~\cite{Sankowski05}, such a matrix equals the inverse
of the matrix $I-x\cdot A$ modulo $x^{h+1}$, where $A$ is the adjacency matrix of $G$.
Since the single-edge updates to $G$ correspond to single-element updates to $I-x\cdot A$,
maintaining the required path counts reduces to maintaining matrix inverse under
single-element updates. This enables using dynamic matrix inverse data structures~\cite{Sankowski04, BrandNS19, BrandFN22}
with $\Ot(h^2)$ arithmetic operation cost incurred by working with $h$-degree polynomials with $\Ot(h)$-bit coefficients.

For deciding whether a path of a given length exists, it is enough to be able to test whether a path count is non-zero.
By operating modulo a random $\Theta(\log{n})$-bit prime number, whether the individual counts are non-zero can be maintained by operating on $\Ot(1)$-bit polynomial coefficients instead~(see, e.g.,~\cite{KingS02, DemetrescuI00}).
This saves a factor of $h$ in the time cost, but inevitably leads to a Monte Carlo randomized algorithm correct with high probability.
The previous fully dynamic reachability data structures in DAGs~\cite{KingS02, DemetrescuI00} and exact $st$-distances~\cite{Sankowski05, BrandFN22} crucially relied on this idea. In particular, the data structures tailored to reachability~\cite{KingS02, DemetrescuI00} do not keep track of path lengths and do not differentiate between short and long paths; they do apply path counting though (with counts of order $n^n$) and thus this trick saves a factor $n$ there.

\paragraph{Previous dynamic distances data structures} The fully dynamic $st$-distances data structures~\cite{Sankowski05, BrandFN22} apply path counting in combination with the well-known hitting set trick~\cite{UY91}.
If one samples a subset $H\subseteq V$ of size $\Theta((n/h)\log{n})$, then for all $u,v\in V$ with $\dist_G(u,v)\geq h$, some shortest $u\to v$ path contains a vertex of $H$ w.h.p. Another view on this property is that for any $s,t\in V$, some shortest $s\to t$ path in $G$ can be decomposed into subpaths of length at most $h$ with both endpoints
in $H\cup \{s,t\}$. Therefore, given $\dist^h_G(u,v)$ for all $u,v\in H\cup \{s,t\}$,
we can find a shortest $s\to t$ path in $G$ using Dijkstra's algorithm on a digraph on $H\cup \{s,t\}$
with $O(|H|^2)=\Ot((n/h)^2)$ edges. Maintaining the $\Ot((n/h)^2)$ pairwise
$h$-bounded distances between the vertices $H\cup \{s,t\}$ requires at least $\Ot((n/h)^2\cdot h)=\Ot(n^2/h)$ bits when using path counting modulo a prime.

This simple approach seems hopeless if we either wanted to make it work for the more
general single-source distances maintenance, or to have a deterministic dynamic $st$-distance data structure. In the former case, we would need to maintain $h$-bounded distances between $O(n^2/h)$ pairs ${(H\cup\{s\})\times V}$. But the current path counting techniques assume at least $\Theta(h)$ ``arithmetic'' overhead for maintaining one value $\dist^h_G(u,v)$ even when counting modulo a random prime. This makes the single-source case inherently quadratic using this approach.
Similarly, deterministic maintenance of a single value $\dist^h_G(u,v)$ has $\widetilde{\Theta}(h^2)$ ``arithmetic'' overhead, which makes the approach
of~\cite{Sankowski05, BrandFN22} inapplicable not only for maintaining $st$-distance deterministically,
but also for fully dynamic deterministic $st$-reachability which is likely an easier problem.

\paragraph{Our $h$-bounded distances data structure}
At the end of the day, we do not overcome the aforementioned inherent arithmetic overheads
of either $\widetilde{\Theta}(h)$ (randomized) and $\widetilde{\Theta}(h^2)$ (deterministic). Instead, we rely, in a crucial way, on the ability of the dynamic matrix inverse data structure~\cite{Sankowski04} to maintain an arbitrary \emph{dynamic} $k$-cell \emph{submatrix of interest} of the inverse, possibly with no structure whatsoever, using a subquadratic number of field operations, unless $k=n^{2-o(1)}$ or significantly more than $\Theta(n)$ cells of interest are added per element update.
Specifically, we show the following:

\begin{restatable}[{Modification to~\cite{Sankowski05}, proof in \Cref{s:algebraic}}]{theorem}{algebraic}\label{t:h-bounded-submatrix-batch}
	Let $G=(V,E)$ be a fully dynamic unweighted digraph. Let $h\in [1,n]$.
	Let ${Y\subseteq V\times V}$ be a dynamic set.
	For any $\mu,\delta \in [0,1]$, there exists a deterministic data structure maintaining the values $\dist_G^h(s,t)$ for all $(s,t)\in Y$
	and supporting:
	\begin{itemize}
		\item single-edge updates to $G$ in $\Ot(h^2\cdot (n^{\omega(1,\mu,1)-\mu}+n^{1+\mu}+|Y|))$ worst-case time,
		\item batches of at most $n^\delta$ \emph{vertex updates} to $G$ in $\Ot(h^2n^{\omega(1,\delta,1)})$ worst-case time,
		\item single-element insertions to $Y$ in $\Ot(h^2n^\mu)$ time,
		\item removing elements from $Y$ in $O(1)$ time.
	\end{itemize}
\end{restatable}
Crucially, for small-polynomial values of $h$, sets $Y$ of size $n^{2-O(1)}/h^2$, and appropriately chosen~$\mu$, say $\mu=1/2$, the single-edge update time bound of this data structure is subquadratic in $n$.

To the best of our knowledge, previous algebraic dynamic graph data structures required
only very structured submatrices $Y$ of the inverse to operate. For example, the aforementioned
$st$-distance data structures~\cite{Sankowski05, BrandFN22} use a square submatrix
with rows and columns
${H\cup \{s,t\}}$. In fact,~\cite{BrandFN22} show
a more efficient (than Theorem~\ref{t:h-bounded-submatrix-batch}) data structure
for square inverse submatrix maintenance: the dependence on $n$ in the single-edge update bound
of Theorem~\ref{t:h-bounded-submatrix-batch} is at least $n^{1.529}$ ($n^{1.5}$ if $\omega=2$), whereas~\cite{BrandFN22} achieve the dependence of $n^{1.406}$ ($n^{1.25}$ if $\omega=2$).

\subsection{Combinatorial plug-ins to the algebraic data structure}\label{s:algebraic-combinatorial}
On top of the data structure Theorem~\ref{t:h-bounded-submatrix-batch}, we apply
two crucial black-box enhancements.

\paragraph{Induced subgraph operations or vertex switching}
First, we will need to maintain \linebreak $h$-bounded distances not only in the graph $G$,
but also in an \emph{induced subgraph} $G[W]$, where the set $W\subseteq V$
is also subject to updates. This functionality models temporarily switching some vertices off,
or, equivalently, removing all their incident edges. 

A very simple but powerful observation is that in directed graphs, switching vertices
can be modeled using single-edge updates. This follows by applying a standard reduction
of splitting a vertex $v$ into an in-copy $v^+$ and an out-copy $v^-$, connected
by an edge $v^+v^-$ switched on iff $v\in W$.
As a result, even though switching a vertex off forbids using up to $n$
edges, modeling it costs no more than a single-edge update to the data structure
maintaining the graph. Recall that the single-edge update cost in Theorem~\ref{t:h-bounded-submatrix-batch} is subquadratic if the parameters are appropriately set.

\paragraph{Path reporting via grouping into blocks} A well-known shortcoming
of path-counting \linebreak $h$-bounded distances data structures is that they do not
allow, out-of-the-box, reporting the shortest $h$-bounded paths.
The data structures we develop internally require the path-reporting functionality (against an adaptive adversary) even for maintaining the single-source distances only.

To this end, we make use of a recent technique of~\cite{Karczmarz0S22} used therein for DAGs to convert
a fully dynamic single-source reachability data structure into a fully dynamic single-source \emph{reachability tree} reporting data structure. We apply this technique
in the case of $h$-bounded distances.

The overall idea is as follows. Having fixed a parameter $b\in [1,n]$, the \emph{block size}, the vertex set~$V$ is split into $\ell\leq \lceil n/b\rceil$ blocks $V_1,\ldots,V_\ell$ of no more than $b$ vertices. 
Via vertex-splitting, we can construct $\ell$ related graphs $G_1,\ldots,G_\ell$ such that
a shortest $h$-bounded $s\to t$ path in $G$ is preserved in $G_i$ iff
the penultimate vertex is in the blocks $V_1$ through $V_i$.

Suppose we maintain the $h$-bounded distances data structures (of Theorem~\ref{t:h-bounded-submatrix-batch}) on the graphs $G_1,\ldots,G_\ell$. Then, given a source $s\in V$ and single-source $h$-bounded distances $\dist^h_{G_i}(s,\cdot)$ for all~$i$, we can
identify, for each $t\in T$, at most $b$ potential predecessors $p$ for which
some shortest $h$-bounded $s\to t$ path has a penultimate vertex $p$.
Given those, the single-source shortest $h$-bounded paths tree can be constructed in $O(nb)$ time
by simply running BFS on an $O(nb)$-sized subgraph of $G$.

Overall, if the dynamic $h$-bounded distances data structure has update time $U(n,h)$ and
can compute single-source $h$-bounded distances in $Q(n,h)$ time, then grouping
into blocks yields a data structure with update time $O\left(\frac{n}{b}\cdot U(n,h)\right)$
reporting single-source $h$-bounded shortest paths in $O\left((n/b)Q(n,h)+nb\right)$ time.
In Section~\ref{s:blocks} we also describe a refinement of this reduction that can produce better running times; see Lemma~\ref{l:blocks2}.

Although we need path reporting even internally for the sole purpose of maintaining single-source distances, note that grouping into blocks allows for converting -- in a \emph{deterministic} and black-box way -- any \emph{truly subquadratic} single-source distances data structure into a SSSP-tree reporting data structure. For example, if the single-source distances data structure (i.e., $h=n$) has $U(n,h),Q(n,h)\in O(n^{1.98})$, it is enough to set $h=n$ and $b=n^{0.99}$ to get $O(n^{1.99})$-time SSSP tree reporting.
Therefore, in the remaining part of this outline we only discuss maintaining single-source distances deterministically.

\subsection{Maintaining most $h$-bounded shortest paths under deletions}
Overall, our fully dynamic single-source distances data structure follows
the high-level idea of the known \emph{combinatorial} fully dynamic APSP data structures with subcubic \emph{worst-case} update time, first shown by~\cite{Thorup05} and subsequently developed in~\cite{AbrahamCK17, ChechikZ23, GutenbergW20b, Mao23}.

The data structure operates in phases of $\Delta= n^\delta$ updates, where $\delta\in (0,1)$.
At any moment of time, let $G$ denote the current graph (with all updates issued so far applied), and let $G_0$ reflect how $G$ looked like when
the current phase started.
While the phase proceeds, let $D$ denote the set of $\leq 2\Delta$ \emph{affected vertices}, i.e.,
those constituting the endpoints of all the edges inserted or deleted in the current phase.
Note that we have $G-D=G_0-D\subseteq G$ at any time.

Each phase starts with a relatively costly preprocessing step
whose purpose is to find a certain collection
of $O(n^2)$ paths $\Pi$ in $G_0$ accomplishing the following informal goals:
\begin{enumerate}[(a)]
  \item $\Pi$ contains all-pairs $h$-bounded shortest paths in a \emph{large} subgraph of $G_0-D$.
  \item removing $D$ from $G$ keeps \emph{most} of the paths in $\Pi$ intact.
\end{enumerate}
More precisely, we use the following lemma inspired by the preprocessing in~\cite{GutenbergW20b}.

\newcommand{\dstree}{\mathcal{T}}
\newcommand{\dsy}{\mathcal{A}}

\begin{restatable}{lemma}{lcollectionthres}\label{l:collection-thres}
  Let $G=(V,E)$ be a directed graph. Let $h,\tau\in [1,n]$.
  Assume one has constructed a \emph{deterministic} data structure storing $G$ and a subset $W\subseteq V$ initially equal to $V$, and supporting:
  \begin{enumerate}[(1)]
    \item updates removing a vertex from $W$,
    \item queries computing shortest $h$-bounded paths in $G[W]$ from a query vertex $s$ to \emph{all} $v\in V$.
  \end{enumerate}
  Suppose one can perform updates in $U(n,h)$ \emph{worst-case} time, and queries in $Q(n,h)$ time. 

  In $O(n^2h + (n/\tau)\cdot U(n,h) + n\cdot Q(n,h))$ time one can compute a set $C\subseteq V$ of
  size $O(n/\tau)$, and a collection of paths $\Pi$ in $G$, containing at most one path $\pi_{s,t}$
  for all $s,t\in V$, such that:
  \begin{enumerate}[(a)]
    \item For any $s,t\in V$, if $\dist_{G-C}^h(s,t)\neq\infty$, then $\pi_{s,t}\in \Pi$ exists and is an $s\to t$ path in $G$ of length at most $\dist_{G-C}^h(s,t)$.
    \item For any $u\in V$, there exist at most $O(hn\cdot \tau)$ paths $\pi_{s,t}\in \Pi$ with $u\in V(\pi_{s,t})$.
  \end{enumerate}
  After the algorithm completes, the assumed data structure is reverted back to the state when $W=V$.
\end{restatable}
We call $\tau$ and $C$ in Lemma~\ref{l:collection-thres} the \emph{congestion threshold} and the \emph{congested subset}, respectively.
The preprocessing of Lemma~\ref{l:collection-thres}, applied to $G_0$ with a small-polynomial~$\tau$ and $h$, achieves the goal~(a) since the obtained collection $\Pi$ contains paths at least as
short as shortest $h$-bounded paths in $G_0-(C\cup D)$. Note that $G_0-(C\cup D)$ forms a
\emph{large} subgraph of $G_0-D$ because the obtained set of congested vertices $C$ has sublinear size.
It also achieves the goal~(b) since the $O(hn\cdot \tau\cdot \Delta)$ paths in $\Pi$ containing
a vertex from $D$ (let us denote them by $\Pi_D$) constitute a relatively small portion of the entire collection
$\Pi$ containing typically $\Theta(n^2)$ paths.

Using purely combinatorial data structures, the preprocessing of Lemma~\ref{l:collection-thres} takes $\Theta(n^3)$ time for dense graphs, which is prohibitive in our case since it is meant to be applied
once per $\Delta$ updates.
However, by plugging in the data structure $\dstree$ of Theorem~\ref{t:h-bounded-submatrix-batch}, extended with induced subgraph operations and path-reporting capabilities (with an appropriately chosen block size $b\in [1,n]$, as explained in Section~\ref{s:algebraic-combinatorial}), we can manage to compute the pair $(\Pi,C)$ in subcubic time upon the start of the phase.
For large enough $\Delta$, this incurs a subquadratic \emph{amortized} update cost.
However, we remark that the heavy super-quadratic preprocessing (which constitutes the only source of amortization here) is scheduled precisely every $\Delta$ updates and consequently the data structure can be deamortized using a standard technique; see, e.g.,~\cite{AbrahamCK17, BrandNS19, Thorup05}.

\subsection{Maintaining $h$-bounded distances beyond the preprocessed pairs}
The preprocessing of Lemma~\ref{l:collection-thres} applied at the beginning of a phase guarantees that $h$-bounded distances in $G_0-(C\cup D)=G-(C\cup D)$ are maintained for all $(s,t)$ such that either $\dist_{G_0-C}(s,t)=\infty$ holds, or $\pi_{s,t}\in\Pi\setminus \Pi_D$,
that is, we have $\dist_{G_0-(C\cup D)}^h(s,t)=|\pi_{s,t}|$.

We delegate maintaining $h$-bounded distances (in $G$) either corresponding to paths going \emph{through} $C\cup D$ and/or for the \emph{affected pairs} $(s,t)$ such that $\pi_{s,t}\in \Pi_D$ to a separate data structure
$\dsy$ of Theorem~\ref{t:h-bounded-submatrix-batch} maintaining the current graph $G$.
More specifically, we define the set $Y\subseteq V\times V$ (identified with the submatrix $Y$ of interest in $\dsy$, see Theorem~\ref{t:h-bounded-submatrix-batch})
to contain, at any time:
\begin{itemize}
  \item all pairs $((C\cup D)\times V) \cup (V\times (C\cup D))$,
  \item all pairs $(s,t)$ such that $\pi_{s,t}\in \Pi_D$.
\end{itemize}
Note that the size of $Y$ is bounded by $O(n^2/\tau+n\Delta+\Delta hn\tau)$ at all times,
which incurs a per-update cost of $O(h^2n^2/\tau+\Delta h^3n\tau)$
for maintaining $\dist_G^h(s,t)$ for all $(s,t)\in Y$ in the data structure~$\dsy$.
Moreover, we crucially rely on the low-recourse property of the set $Y$: it grows by at most $O(hn\tau)$
elements (worst-case) per update, as $D$ grows by at most two elements per update.
This allows to augment the submatrix of interest in $\dsy$ in $O(h^3n^{1+\mu}\tau)$ time per update.
We point out that it is the affected pairs that cause the submatrix
of interest $Y$ to be of a rather arbitrary shape beyond our control;
the remaining part of $Y$ is composed of some $O(n/\tau+\Delta)$ full rows and columns.

\subsection{Speeding up queries by hitting set-based sparsification}
Let $s\in V$ be some query vertex. Suppose we want to compute single-source
distances from $s$ in $G$.
First of all, using the data structure $\dsy$ we can temporarily add
all pairs $(s,t)$, $t\in V$, to the set $Y$ so that we can also access
single-source $h$-bounded distances from $s$.
This costs $\Ot(h^3n^{1+\mu})$ time.

With such an augmented set $Y$, one can prove (see Lemma~\ref{l:correct}) that for each $t$, there exists a shortest path $P=s\to t$ in~$G$ obtained by stitching paths of one of the following kinds:
\begin{enumerate}[(i)]
  \item a shortest $h$-bounded path $u\to v$, where $(u,v)\in Y$,
  \item a path $\pi_{u,v}\in \Pi\setminus \Pi_D$ such that $|\pi_{u,v}|\geq h/2$.
\end{enumerate}
Consequently, we could build a graph $X$ with edges $(u,v)\in Y$ of weight $\dist_G^h(u,v)$ (maintained by~$\dsy$) and edges $(u,v)$ for all non-affected pairs with $\pi_{u,v}\in \Pi\setminus \Pi_D$, and run Dijkstra's algorithm on $X$ to obtain single-source distances from $s$ in $G$.
The number of edges of the former kind is $O(n^2/\tau+hn\Delta\tau)$, i.e., subquadratic.
However, the number of edges of the latter kind can be still $\Theta(n^2)$.
Fortunately, these edges represent relatively long $\Theta(h)$-hop paths, and this can be exploited as follows.
During the preprocessing at the beginning of a phase, after computing the collection $\Pi$,
we compute an $\Ot(n/h)$-sized hitting set $H$ of the paths in $\Pi$
composed of at least $h/2$ edges. This can be done deterministically in $\Ot(n^2h)$
time using a folklore greedy algorithm. For each of these paths $\pi_{u,v}$,
we pick one $z_{u,v}\in H\cap V(\pi_{u,v})$ as a representative of the hitting set
for that path.

To be able to efficiently compute single-source distances upon query,
we sparsify the lengths of the $O(n^2)$ paths of type~(ii) using only $\Ot(n^2/h)$ edges.
For each $(z,v)\in H\times V$, we use an edge $zv$ of weight equal to the
minimum length of the $z\to v$ suffix of a path $\pi_{u,v}\in \Pi\setminus \Pi_D$ such that $z_{u,v}=z$. Similarly, for each $(u,z)\in V\times H$, 
we use an edge $uz$ of weight equal to the minimum length of the $u\to z$ prefix
of a path $\pi_{u,v}\in\Pi\setminus \Pi_D$ with $z_{u,v}=z$.
These auxiliary edges can be efficiently maintained subject to insertions to $D$
(equivalently, deletions from $\Pi\setminus \Pi_D$).
Using the auxiliary edges in place of explicit edges for all the non-affected
paths decreases the size of the graph~$X$ to $O(n^2/\tau+hn\Delta\tau+n^2/h)$
which is subquadratic for a proper choice of parameters.

We remark that our usage of hitting sets is conceptually different compared to previous dynamic
shortest paths algorithms.
There, hitting sets are typically used to hit long-hop paths in the \emph{current} graph,
most commonly by using randomization to ensure a once-sampled hitting
set is valid for $\poly(n)$ (future) versions of the graph (e.g., in~\cite{AbrahamCK17, Sankowski05, BrandFN22}).
In our case, we use it merely for the sake of sparsification of some statically precomputed
set of paths and do not require maintaining the current all-pairs $h$-bounded distances between the vertices of the hitting set.

\subsection{Fully dynamic $(1+\eps)$-approximate APSP with non-trivial worst-case update bound and optimal path-reporting}
In order to obtain the deterministic all-pairs data structures supporting optimal-time path reporting, we build upon the framework that we have used for single-source distances.
First of all, to be competitive with the known all-pairs reachability and APSP data structures, we need to support more complex \emph{vertex updates}. Fortunately, the SSSP framework can be seen to
allow decremental vertex updates within a phase, and only relies on the single-edge updates assumption (i.e., that every updated edge has \emph{both} endpoints in $D$)
in the query procedure.
Moreover, the data structure of Theorem~\ref{t:h-bounded-submatrix-batch} supports batch-vertex updates, which we leverage when advancing to the next phase.

In our all-pairs data structures, the single-source query procedure is replaced by \emph{explicitly rebuilding}
a collection $\Pi'$ (containing, in particular, the unaffected paths $\Pi\setminus \Pi_D$) of $h$-bounded paths $\pi'_{u,v}$ for the $O(hn\Delta\tau)$ affected pairs $(u,v)$ (with $\pi_{u,v}\in \Pi_D$) satisfying similar guarantees wrt. the $h$-bounded distances in $G-(C\cup D)$.
This is achieved via another application of the data structure of Theorem~\ref{t:h-bounded-submatrix-batch}, enhanced as described in Section~\ref{s:algebraic-combinatorial}.
Importantly, this allows for optimal-time reporting of $h$-bounded paths
in $G$ that are at least as short as those appearing in $G-(C\cup D)$.

With the shortest $h$-bounded paths $\Pi'$ improving upon those in $G-(C\cup D)$ rebuilt,
we deterministically compute an $\Ot(n/h)$-sized hitting set $H$ of
the $\Theta(h)$-hop paths in $\Pi'$.
One can show that for every pair $(s,t)$, either $\pi'_{s,t}\in \Pi'$
is a shortest $s\to t$ path in $G$, or there is a shortest path $P=s\to t$
in $G$
with a vertex in $C\cup D\cup H$ appearing
on each $h$-hop segment of $P$.
If $|C\cup D\cup H|\leq n^\alpha$, for some $\alpha<1$,
then this allows computing all-pairs shortest paths of the latter kind
using $O(\log{n})$ min-plus products of matrices whose at least one dimension
is bounded by $O(n^\alpha)$.
Such products can be computed polynomially faster than products of full $n\times n$
matrices, even without resorting to non-trivial rectangular matrix multiplication algorithms~\cite{GallU18, HuangP98}.
Sadly, in our case, using exact min-plus product for this task does not yield
an improved worst-case update bound over the known exact APSP data structures~\cite{AbrahamCK17, GutenbergW20b} with subcubic worst-case update bounds.

However, if we use the $(1+\eps)$-approximate min-plus product~\cite{Zwick02} (where $\eps\in (0,1]$)
we can perform this computation approximately
in sub-$n^\omega$ time $\Ot(n^{\omega(1,\alpha,1)}/\eps)$.
The approximate min-plus product can also produce matrix multiplication
witnesses within the same bound deterministically.
Those, in turn, can be used to recover, for any query pair $(s,t)$, some
$s\to t$ path of length at most $(1+\eps)\cdot \dist_G(s,t)$
in optimal $O(\dist_G(s,t))$ time.

For the special case of all-pairs reachability (i.e., any
finite approximation of shortest paths), another application
of matrix multiplication witnesses can yield a stronger
path-reporting interface. Namely, within $O(n^{\omega(1,\alpha,1)})$ time
we can also explicitly recompute all-pairs reachability trees.

Finally, we remark that, at the end of the day, due to the overheads of the preprocessing and rebuilding
short-hop paths, the worst-case update time of this data structure
is faster than the recompute-from-scratch bound $\Ot(n^\omega/\eps)$ only
conditionally, if $\omega\geq 2.12$.

\subsection{Deterministic fully-dynamic transitive closure with an unconditional sub-$n^\omega$ update bound}
For the fully dynamic transitive closure problem, we show that the deterministic sub-$n^\omega$ worst-case update time can be attained unless $\omega=2$.
We achieve that using a different and simpler approach.

Recall that a randomly sampled hitting set of size $\Ot(n/h)$ hits some polynomial
number of fixed $\geq h$-hop \emph{shortest paths} with high probability.
Roughly speaking, given any such hitting set $H$ of size at most $n^\alpha$, $\alpha<1$,
we can separately compute reachability through $H$ using rectangular
matrix multiplication in $O(n^{\omega(1,\alpha,1)})$ time, and delegate
maintaining all-pairs $h$-bounded distances to the deterministic data structure of Theorem~\ref{t:h-bounded-submatrix-batch} with $\Ot(h^2n^2)$ worst-case update time per vertex update.
However, whereas such a hitting set $H$ is trivial to obtain using randomization,
it is not clear whether it can be computed deterministically
even in $O(n^\omega)$ time.

Our main idea is to use a weaker -- both qualitatively and quantitatively -- notion of a hitting set of $\geq h$-hop paths.
First of all, for the purposes of reachability, it is enough that for any $s,t\in V$ with $\dist_G(s,t)\geq h$, the hitting set
$H$ hits an $h$-edge prefix of \emph{some} (not necessarily shortest) $s\to t$ path.
Such a qualitatively weaker hitting set of size $\Ot(n/h)$ can be computed based on all-pairs reachability trees, computable in $\Ot(n^\omega)$ deterministic time~\cite{AlonGMN92}.

Second, a larger (but still sublinear) hitting set with the same guarantee can be computed more efficiently. Specifically, we show that if we are willing
to accept a hitting set of size $\Ot((n/d)\cdot (n/h))$, for a parameter $d\in [1,n]$, where $d\cdot h=\Omega(n)$, then it can be constructed deterministically in $\Ot(n^2+nd^{\omega-1})$ time. We show that for any value of $\omega>2$,
the parameters $h$ and $d$ can be set so that the worst-case update time
of the data structure is polynomially smaller than $n^\omega$.

A drawback of this data structure is that it cannot report
$s\to t$ paths in time linear in the length of the reported path, let
alone in optimal $O(\dist_G(s,t))$ time.
That being said, paths can still be reported in $O(n)$ time
using a path reporting procedure of~\cite{Karczmarz0S22}.

\subsection{Randomized improvements}
Using randomization, we can slightly improve the update bounds
of both the SSSP data structure and the $(1+\eps)$-approximate
APSP data structure. First of all, in both cases, using
path counting modulo a prime shaves an $h$-factor from
most of the terms in the update bounds.
Then, the obtained data structures become Monte Carlo randomized
(correct with high probability) but can still work against
an adaptive adversary since the path counts are never revealed.

With randomization, the update/query bound of the fully dynamic SSSP data structure can be improved even further.
Using standard $\Ot(n/h)$-sized
randomly sampled hitting sets, one can
replace the $O(n^2)$-sized path collection of Lemma~\ref{l:collection-thres}
with a different $O(n^2/h)$-sized collection
as used in the fully dynamic APSP data structure of~\cite{AbrahamCK17} (see Lemma~\ref{l:collection-rand}).
Overall, the powerful sampled hitting set makes it possible
to avoid two components crucial in the deterministic data structure:
(a) the explicit sublinear congested subset, and (b) sparsifying
the surviving paths from the collection (since the collection has subquadratic size).
It is worth noting that since the distances output by the data
structure are exact (and thus uniquely determined by the graph),
the data structure still works against an adaptive adversary.
See Section~\ref{s:randomized} for details.

\section{Reporting paths via grouping into blocks}\label{s:blocks}
In this section we state a reduction, based on a technique of~\cite{Karczmarz0S22} (with some adjustments and refinements of ours), for converting data structures
maintaining distances
into path-reporting algorithms. The reduction is based on grouping vertices into blocks of size $b$, where $b$ is a parameter in $[1,n]$.
\cite{Karczmarz0S22} used
this technique for converting a subquadratic dynamic single-source reachability data structure~\cite{Sankowski04} into a subquadratic (albeit polynomially slower)
fully-dynamic
single-source \emph{reachability tree reporting} data structure against an adaptive adversary. We apply the technique to $h$-bounded shortest paths instead. Moreover, by using vertex blocks in a slightly different way, we obtain a generalization (Lemma~\ref{l:blocks2}) that is tailored to our use cases.
The proofs from this section can be found in Appendix~\ref{s:omitted}.

\newcommand{\qpair}{Q_{\mathrm{pair}}}
\newcommand{\qsource}{Q_{\mathrm{source}}}
\newcommand{\upair}{U_{\mathrm{pair}}}
\newcommand{\usource}{U_{\mathrm{source}}}
\newcommand{\ipair}{I_{\mathrm{pair}}}
\newcommand{\isource}{I_{\mathrm{source}}}

  Let $G=(V,E)$ be a directed graph. Let $h\in [1,n]$. Suppose in $\isource(n,h)$ time one can construct a data structure $\mathcal{D}(G)$
  maintaining a graph $G$ under some type of updates (say, single-edge or vertex updates) and supporting single-source queries computing $\dist^h_G(s,t)$ for
  a given query source~$s$ and  all ${t\in V}$. Suppose $\mathcal{D}(G)$ has $\usource(n,h)$ \emph{worst-case} update time and
  $\qsource(n,h)=\Omega(n)$ query time,
  for some non-decreasing functions $\isource,\usource,\qsource\in O(\poly(n,h))$.

  Our goal is to have a data structure capable of handling the following \emph{submatrix predecessor queries}.
  Given query sets $S\subseteq V$ and $Y\subseteq S\times V$,
  compute for each $(s,t)\in Y$ the predecessor vertex on some shortest $h$-bounded $s\to t$ path (if one exists).
  
  Observe that an immediate application of a submatrix predecessor query $(S,S\times V)$ is computing single-source $h$-bounded
  shortest paths from $s$ to $v$ for all $s\in S$ and $v\in V$ such that $\dist_G(s,v)\leq h$.

\begin{restatable}{lemma}{lblocks}\label{l:blocks}
  For any $b\in [1,n]$, there exists a dynamic data structure maintaining $G$ and supporting submatrix predecessor queries $(S,Y)$
  in
  \begin{equation*}
    O\left(\frac{n}{b}\cdot |S|\cdot\qsource(n,h)+|Y|\cdot b\right)
  \end{equation*}
  time.
  The update and initialization times are 
  $O\left(\frac{n}{b}\cdot \usource(n,h)\right)$ and $O\left(\frac{n}{b}\cdot \isource(n,h)\right)$ resp.
\end{restatable}
Now suppose that in addition to $\mathcal{D}(G)$, there is a data structure $\mathcal{D}'(G)$ supporting 
queries about $\dist^h_G(s,t)$ for arbitrary pairs $s,t\in V$ in $\qpair(n,h)$ time
and maintaining a graph $G$ subject to the same type of updates as $\mathcal{D}(G)$ with worst-case update
time $\upair(n,h)$ and initialization time $\ipair(n,h)$.
We can then improve the bound of Lemma~\ref{l:blocks}.

\begin{restatable}{lemma}{lblockstwo}\label{l:blocks2}
  For any $b\in [1,n]$, there exists a dynamic data structure maintaining $G$ and supporting submatrix predecessor queries $(S,Y)$
  in
  \begin{equation*}
    O\left(|S|\cdot \qsource(n,h)+|Y|\cdot \qpair(n,h)\cdot \log{n}+|Y|\cdot b\right)
  \end{equation*}
  time.
  The update and initialization times are
  $O\left(\usource(n,h)+\frac{n}{b}\cdot \upair(n,h)\right)$ and \linebreak $O\left(\isource(n,h)+\frac{n}{b}\cdot \ipair(n,h)\right)$ respectively.
\end{restatable}

\begin{remark}
  The data structures of Lemmas~\ref{l:blocks}~and~\ref{l:blocks2} are deterministic if the assumed
  data structures $\mathcal{D}(G)$ and $\mathcal{D}'(G)$ are deterministic. 
  More generally, if the data structures $\mathcal{D}(G)$ and $\mathcal{D}'(G)$ produce correct
  answers within the assumed bounds against an adaptive adversary, so do the obtained submatrix predecessor
  data structures.
\end{remark}

\paragraph{Supporting induced subgraph queries.}
Later on we will also require reporting shortest $h$\nobreakdash-bounded paths in an induced subgraph
$G[W]$ of a fully dynamic $G$, where $W\subseteq V$ is also dynamic. This can be achieved via a simple yet powerful reduction described
in the following.
\begin{restatable}{lemma}{linducedsubgraph}\label{l:induced-subgraph}
  Let $h\in [1,n]$.
  Let $\mathcal{D}(G)$ be a data structure maintaining a fully dynamic unweighted digraph $G=(V,E)$ subject to single-edge updates in $U(n,h)$ time per update,
  and possibly also vertex updates, and supporting submatrix predecessor queries $(S,Y)$
  in $Q(|S|,|Y|)$ time.

  Let $W\subseteq V$ be a dynamic subset.
  Then there exists a data structure $\mathcal{D}'$ supporting the same set of updates to $G$
  and supporting submatrix predecessor queries $(S,Y)$ in the induced subgraph $G[W]$
  in $O(Q(|S|,|Y|))$ time.
  $\mathcal{D}'(G)$ processes updates to $G$ within asymptotically the same time 
  that $\mathcal{D}(G)$ needs for processing the respective updates to $G$.
  Moreover, a single-element update to~$W$ can be performed in $O(U(n,h))$ time.
\end{restatable}
The details of the reduction can be found in Appendix~\ref{s:omitted}.

\section{Maintaining most $h$-bounded shortest paths under deletions}\label{s:collection}
In this section we show that with the help of algebraic fully dynamic distance data structures storing $G$, one can preprocess $G$ in subcubic
time so that most of the $h$-bounded shortest paths in a very
large subgraph of $G$ are intact by any sufficiently small adversarially chosen
set of deleted vertices.

We start with the following lemma inspired by the construction of the fully dynamic APSP algorithms of~\cite{GutenbergW20b}.

\lcollectionthres*
\begin{proof}
  Initially we put $C=\emptyset$ and $\Pi=\emptyset$. We will gradually extend $C$ and $\Pi$ to make the condition (a) satisfied for more and more pairs $(s,t)$.
  We maintain the invariant that the set $W$ in the assumed data structure equals $V\setminus C$.
  That is, whenever a new element is added to $C$, this is reflected in the data structure by a deletion to $W$.

  Moreover, let $\alpha(v):=|\{\pi\in \Pi:v\in V(\pi)\}|$.
  Initially all values $\alpha(\cdot)$ are zero.
  We maintain an invariant that $\alpha(v)\leq 2nh\tau$ for all $v\in V$.
  By this invariant, (b) will follow when we are done.

  We process the sources $s\in V$ one by one. For a fixed source $s$, we 
  start by moving to $C$ all the vertices $v\in V$ not yet in $C$
  such that $\alpha(v)>nh\tau$.
  Next, we compute
  shortest $h$-bounded paths from $s$ to all $v\in V$ in $G[V\setminus C]$ in $Q(n,h)$ time
  and add them as the respective paths $\pi_{s,t}$ (if they exist) to $\Pi$.
  We also update the counters $\alpha(\cdot)$.
  Every $\alpha(v)$, for $v\in V\setminus C$, grows by at most $n$ when processing the source $s$, and thus
  afterwards we have $\alpha(v)\leq nh\tau+n\leq 2nh\tau$, so the invariant is satisfied.
  Since the set $C$ may only grow in the process, and thus the values $\dist_{G-C}^h(\cdot,\cdot)$ may only increase, the paths already in $\Pi$ satisfy (a).
  Moreover, after processing the source $s$, we ensure that the condition (a) is satisfied also
  for all pairs $(s,t)$, $t\in V$.

  Finally, to see that $|C|=O(n/\tau)$, observe that for every source $\sum_{v\in V}\alpha(v)$
  grows by at most $n(h+1)$, as each path added to $\Pi$ has at most $h+1$ vertices.
  As a result, at the end $\sum_{v\in V}\alpha(v)\leq n^2(h+1)\leq 2n^2h$.
  Since each $c\in C$ satisfies $\alpha(c)>nh\tau$, we obtain $|C|\leq 2n/\tau$.

  Once we have completed computing $(C,\Pi)$, we revert all the changes
  to the assumed data structure, so that $W$ is again equal to $V$. This can be done
  within the same bound as the the described computation since
  the assumed update/query bounds are worst-case.
\end{proof}

We will repeatedly use the following deterministic algebraic data structure maintaining \linebreak
$h$-bounded distances of a fully dynamic graph $G$ subject to single-edge updates and/or vertex updates.
This data structure encapsulates and extends various earlier developments of~\cite{Sankowski04,Sankowski05,BrandFN22}.
The full proof is given in \Cref{s:algebraic}.

\algebraic*

\paragraph{The path-reporting data structure $\dstree$.} Let us apply the construction of Lemma~\ref{l:blocks2} on top of the the data structure of \Cref{t:h-bounded-submatrix-batch},
and subsequently augment it with induced subgraph operations as described in Lemma~\ref{l:induced-subgraph}.
This way, we obtain a data structure $\dstree$
required by Theorem~\ref{l:collection-thres}.

More specifically, in terms of Lemma~\ref{l:blocks2}, the data structure of Theorem~\ref{t:h-bounded-submatrix-batch}
satisfies
\begin{align*}
  \qsource(n,h)&=\Ot(h^2n^{1+\mu}),\\
  \qpair(n,h)&=\Ot(h^2n^\mu),\\
  \usource(n,h)=\upair(n,h)&=\Ot(h^2(n^{1+\omega(1,\mu,1)-\mu}+n^{1+\mu})).
\end{align*}
Consequently, for any chosen $b\in [1,n]$, the data structure 
$\dstree$ supports single-edge updates to $G$ and switching off a vertex (via its removal from $W\subseteq V$) in worst-case time
\begin{equation*}
  U(n,h)=O\left(\frac{h^2n^{1+\omega(1,\mu,1)-\mu}}{b}+\frac{h^2n^{2+\mu}}{b}\right).
\end{equation*}
It also supports reporting $h$-bounded shortest path from any single source $s\in V$ in
\begin{equation*}
  Q(n,h)=O\left(h^2n^{1+\mu}\log{n}+nb\right)
\end{equation*}
time, by issuing a submatrix predecessor query $(\{s\},\{(s,v):v\in V\})$.
By plugging these bounds in Lemma~\ref{l:collection-thres} and noting
that $n\cdot Q(n,h)=\Omega(n^2h)$, we conclude
with the following theorem which formally captures what we mean
by preserving \emph{most} $h$-bounded shortest paths in a large subgraph of $G$ under a bounded number of deletions.

\begin{theorem}\label{t:collection}
  Let $h,\tau,b\in [1,n]$ and $\mu\in [0,1]$.
  Assume $G$ is maintained in a data structure $\dstree$ of \Cref{t:h-bounded-submatrix-batch} extended
  as described in Lemmas~\ref{l:blocks2}~and~\ref{l:induced-subgraph} (for the given $b$). Then,
  in time
  \begin{equation*}
    \Ot\left(\frac{h^2n^{2+\omega(1,\mu,1)-\mu}}{b\tau}+\frac{h^2n^{3+\mu}}{b\tau}+n^2b\right),
  \end{equation*}
  one can deterministically compute a pair $(C,\Pi)$, where $C\subseteq V$ is of
  size $O(n/\tau)$, and $\Pi$ is a collection of paths in $G$, containing at most one path $\pi_{s,t}$
  for all $s,t\in V$, such that:
  \begin{enumerate}[(a)]
    \item For any $i=1,\ldots,k$ and $s,t\in V$, if $\dist_{G-C}^h(s,t)\neq\infty$, then $\pi_{s,t}\in \Pi$ exists and is an $s\to t$ path in $G$ of length at most $\dist_{G-C}^h(s,t)$.
    \item For any $u\in V$, there exist at most $O(hn\cdot \tau)$ paths $\pi_{s,t}\in \Pi$ with $u\in V(\pi_{s,t})$.
  \end{enumerate}
\end{theorem}
In our concrete applications, we will set $b$ to be slightly sublinear and both $h$ and $\tau$
to be small-polynomial so that $n-o(n)$ vertices of $G$ remain after removing the set $C$,
and a small-polynomial $\Delta$ adversarial vertex deletions
can destroy at most $O(\Delta h\tau n)=o(n^2)$ paths in $\Pi$.

\section{Fully dynamic single-source shortest paths}\label{s:sssp}
In this section we describe our deterministic fully dynamic SSSP data structure for unweighted graphs
with truly subquadratic update and query time.

\paragraph{Phases and global data structures.} Let $h,\tau,\Delta,b\in [1,n]$ and $\mu,\nu\in [0,1]$ be parameters to be set later.
The data structure operates in phases of $\Delta$ \emph{single-edge} updates. It also maintains
two \emph{global} dynamic data structures shared between the phases:
\begin{itemize}
\item A data structure $\dsy$ of \Cref{t:h-bounded-submatrix-batch} with $\mu:=\nu$ always storing the
  current graph $G$. Every edge update issued to $G$ is immediately passed to $\dsy$. 
  The dynamic set $Y\subseteq V\times V$ of vertex pairs of interest in $\dsy$, to be defined later, is initialized
    at the beginning of
    a phase. Moreover, the set $Y$ grows by at most $O(\tau hn)$ pairs per update in the worst-case.
    However, we will guarantee that the size of $Y$ is bounded by $O(n^2/\tau+\Delta\tau hn)$ at all times.
  By \Cref{t:h-bounded-submatrix-batch}, the worst-case update time of $\dsy$ is:
  \begin{equation*}
    \Ot(h^2\cdot (n^{\omega(1,\nu,1)-\nu}+n^{1+\nu}+|Y|)+h^2 \cdot \tau hn^{1+\nu})=\Ot\left(h^2n^{\omega(1,\nu,1)-\nu}+\frac{h^2n^2}{\tau}+h^3\tau n\cdot (\Delta+n^\nu)\right).
  \end{equation*}
\item A data structure $\dstree$, as required by Theorem~\ref{t:collection} and described in Section~\ref{s:collection},
  with parameters $h,b,\mu$.
  Every edge update issued to $G$ is immediately passed to $\dstree$. 
  Recall that the cost of updating $\dstree$ incurred by a single update to $G$ is:
  \begin{equation*}
  U(n,h)=O\left(\frac{h^2n^{1+\omega(1,\mu,1)-\mu}}{b}+\frac{h^2n^{2+\mu}}{b}\right).
  \end{equation*}
    While the phase proceeds, the set $W$ in the data structure $\dstree$ (describing the induced subgraph of interest)
    equals $V$
    and it is only manipulated when a phase initializes.
  Similarly, queries to $\dstree$ will only be issued at the beginning of a phase.
\end{itemize}

\paragraph{Initializing a phase.}
For a fixed current phase, let $G_0$ denote the graph $G$ at the beginning of that phase.
We always use $G$ to refer to the \emph{current graph}, i.e.,
with all updates issued so far applied.
Moreover, when a phase proceeds, let $D\subseteq V$ be the set of endpoints of the edges
affected in the current phase, that is, the edges that have been inserted or deleted during the current phase.
Note that by the definition of $D$, we have $|D|\leq \Delta$ and $G_0-D=G-D\subseteq G$ at all times.

At the beginning of a phase, we apply Theorem~\ref{t:collection} to obtain a pair $(C,\Pi)$,
where $C\subseteq V$ is a ``congested subset'' of size $O(n/\tau)$, and $\Pi$ is a collection of paths in $G_0$.
By Theorem~\ref{t:collection}, the amortized cost of such a preprocessing step, happening once per
$\Delta$ updates, is:
\begin{equation}\label{eq:collection-amortized}
  \Ot\left(\frac{n^2h}{\Delta}+\frac{h^2n^{2+\omega(1,\mu,1)-\mu}}{b\tau \Delta}+\frac{h^2n^{3+\mu}}{b\tau \Delta}+\frac{n^2b}{\Delta}\right).
\end{equation}

We next deterministically compute an $\Ot(n/h)$-sized hitting set $H$ of 
those paths $\pi_{s,t}\in \Pi$ \emph{that have at least $\lfloor h/2\rfloor$ edges}.
For this, we use the following folklore result (see, e.g.,~\cite{Zwick02}).
\begin{lemma}\label{l:hitting}
  Let $X$ be a ground set of size $n$ and let $\mathcal{Z}$ be a family of subsets of $X$, each with at least $k$ elements.
  Then, in $O\left(\sum_{Z\in\mathcal{Z}}|Z|\right)$ time one can deterministically compute a \emph{hitting set} $R\subseteq X$
  of size $O((n/k) \cdot\log{n})$ such that $R\cap Z\neq\emptyset$ for all $Z\in \mathcal{Z}$.
\end{lemma}
The total size of $\Pi$ is $O(n^2h)$ and each path in $\Pi$ is a simple path; therefore, a desired set $H$
can be computed in $O(n^2h)$ time.
Again, this incurs an amortized cost of $O(n^2h/\Delta)$ per update.
which is asymptotically the same as the first term in~\eqref{eq:collection-amortized}.

With the hitting set $H$ in hand, for each $\pi_{s,t}$ with $|\pi_{s,t}|\geq \lfloor h/2\rfloor$, we arbitrarily pick a
single representative vertex $\beta(\pi_{s,t}) \in H\cap V(\pi_{s,t})$.
If $|\pi_{s,t}|<\lfloor h/2 \rfloor$, we put $\beta(\pi_{s,t}):=\perp$.

\newcommand{\fromh}[1]{{\ensuremath{\Psi^{\mathrm{from}}_{#1}}}}
\newcommand{\toh}[1]{{\ensuremath{\Psi^{\mathrm{to}}_{#1}}}}

Let $\Pi_D$ denote the subset of paths $\pi_{s,t}\in \Pi$ such that $V(\pi_{s,t})\cap D\neq\emptyset$.
Recall that at the beginning of the phase $D=\emptyset$, so $\Pi_D$ equals $\emptyset$ initially.
The set $\Pi_D$ (seen as pointers to the paths in $\Pi$) can be easily extended upon
insertions to $D$ (caused by updates to $G$) in optimal time if additional
$O(n^2h)$ vertex-path pointers are stored for the preprocessed collection $\Pi$.

For a simple path $p$ and $a,b\in V(p)$ let $p[a,b]$ denote the $a\to b$ subpath of $p$.
Throughout a phase, for each $z\in H$ and $v\in V$, we maintain two multisets:
\begin{align*}
  \toh{z}(v)&=\{|\pi_{s,v}[z,v]|:s\in V\text{ such that }\pi_{s,v}\in \Pi\setminus \Pi_D \text{ and }\beta(\pi_{s,v})=z\},\\
  \fromh{z}(v)&=\{|\pi_{v,t}[v,z]|:t\in V\text{ such that }\pi_{v,t}\in \Pi\setminus \Pi_D \text{ and }\beta(\pi_{v,t})=z\}.
\end{align*}
There are $O(|H|n)=\Ot(n^2/h)$ multisets $\toh{z}(v),\fromh{z}(v)$ and their total
size is bounded by $2|\Pi|$, i.e., $O(n^2)$.
At the beginning of the phase, these multisets can be initialized in $O(n^2)$ time,
which, again, is negligible (amortized) to the cost~\eqref{eq:collection-amortized}.
We store them in data structures allowing polylogarithmic deletions
and minimum maintenance, for example balanced binary search trees.

Finally, we are ready to define the set $Y$ of pair $(s,t)$ for which we maintain
values $\dist^h_G(s,t)$ in the data structure $\dsy$.
After the initialization step of a phase completes, at all times $Y$ contains:
\begin{itemize}
  \item all pairs $(u,v)\in (D\times V)\cup (V\times D)$,
  \item all pairs $(u,v)\in (C\times V)\cup (V\times C)$,
  \item all pairs $(s,t)\in V\times V$ such that $\pi_{s,t}\in \Pi_D$.
\end{itemize}
Since $D=\emptyset$ and $\Pi_D=\emptyset$ at the beginning of the phase,
a phase's initialization only potentially empties $Y$ in $\dsy$ and
subsequently inserts all $O(n^2/\tau)$ pairs $(C\times V)\cup (V\times C)$ into $Y$.
By \Cref{t:h-bounded-submatrix-batch}, this incurs an additional amortized update cost of
\begin{equation*}
  \Ot\left(\frac{h^2n^{2+\nu}}{\tau\Delta}\right).
\end{equation*}
\paragraph{Performing updates.}
While the phase proceeds, the sets $D$ and $\Pi_D$, and thus also $Y$, grow.
However, recall from Theorem~\ref{t:collection} that $|D|\leq \Delta$ implies that at all times we have:
\begin{equation*}
  |\Pi_D|=O(\Delta \cdot h\tau n),
\end{equation*}
which also yields the promised $O(n\Delta+\Delta h\tau n+n^2/\tau)=O(n^2/\tau+\Delta h\tau n)$ bound on the size
of the set $Y$. Also, more specifically, at most $O(h\tau n)$ new paths may be inserted
to $\Pi_D$ per update, which is also why $Y$ may grow by at most $O(h\tau n)$ pairs
per update in the worst-case.
The changes to the set $Y$ are immediately passed to the data structure~$\dsy$.

When a new path $\pi_{s,t}$ is added
to $\Pi_D$, and $\beta(\pi_{s,t})=z$, the two multisets
$\toh{z}(t)$ and $\fromh{z}(s)$ have to be updated by issuing a single-element deletion.
Therefore, maintaining the multisets costs $\Ot(hn\tau)$ worst-case time per update,
which is dominated by the update cost of the data structure~$\dsy$.

We are now ready to state the summarized amortized update cost of the data structure:
\begin{equation*}
  \begin{split}
    \Ot\left(\frac{n^2h}{\Delta}+\frac{h^2n(n^{\omega(1,\mu,1)-\mu}+n^{1+\mu})}{b}\cdot\left(1+\frac{n}{\tau\Delta}\right)+\frac{n^2b}{\Delta}+
  h^2n^{\omega(1,\nu,1)-\nu}+\frac{h^2n^2}{\tau}\right.\\\left.+h^3\tau n\cdot (\Delta+n^\nu)+\frac{h^2n^{2+\nu}}{\tau\Delta}\right).
  \end{split}
  \end{equation*}

\newcommand{\xdet}{\ensuremath{X}}
\paragraph{Answering queries.}
Suppose a query for distances from a source vertex $s\in V$ is issued.
We first temporarily add all pairs $(s,t)$, where $t\in V$, to the set $Y$ so that this
change is reflected in the data structure~$\mathcal{A}$.
Let $Y_s$ be the set $Y$ augmented this way.
Recall that the augmentation of $Y$ takes $\Ot(h^2n^{1+\nu})$ time, and the change
can be reverted within the same bound.

Next, we build an auxiliary graph $\xdet=(V,F)$ with the following edges:
\begin{itemize}
  \item an edge $uv$ of weight $\dist_G^h(u,v)$ for all $(u,v)\in Y_s$,
  \item an edge $zt$ of weight $\min(\toh{z}(t))$ for all $(z,t)\in H\times V$,
  \item an edge $sz$ of weight $\min(\fromh{z}(s))$ for all $(s,z)\in V\times H$.
\end{itemize}

Recall that the weights of edges of the former type are maintained in the data structure~$\mathcal{A}$.
The graph~$\xdet$ has $O(|H|n+|Y|)=\Ot(n^2/h+\Delta h\tau n+n^2/\tau)$ edges. The query procedure is to simply
run near-linear time Dijkstra's algorithm on the graph $\xdet$. The following lemma establishes correctness.

\begin{lemma}\label{l:correct}
  For all $t\in V$, $\dist_G(s,t)=\dist_{\xdet}(s,t)$.
\end{lemma}
\begin{proof}
  First of all, note that for
  all $uv\in E(\xdet)$, the weight of an edge $uv$ in $\xdet$ equals the weight
  of some $u\to v$ path in $G$. As a result, $\dist_{G}(u,v)\leq \wei_X(uv)$, which establishes $\dist_G(s,t)\leq \dist_{\xdet}(s,t)$ easily.

  Let us now prove that $\dist_{\xdet}(s,t)\leq \dist_G(s,t)$.
  Consider any shortest $s\to t$ path $P$ in $G$.
  First, split~$P$ into maximal segments $P_1\cdot \ldots\cdot P_k$ not passing through any
  of the vertices $C\cup D$ as intermediate vertices. In other words, for each $P_i=u_i\to v_i$, we have $V(P_i)\cap (C\cup D)\subseteq \{u_i,v_i\}$
  and $u_i\in C\cup D\cup \{s\}$.
  We will prove that $\dist_{\xdet}(u_i,v_i)\leq \dist_G(u_i,v_i)$. By summing through $i$, $\dist_{\xdet}(s,t)\leq \dist_G(s,t)$ will follow easily.

  Consider some $P_i$. We have $|P_i|=\dist_{G}(u_i,v_i)$.
  If $|P_i|\leq h$, then $\dist_G(u_i,v_i)=\dist_G^h(u_i,v_i)$
  and by $u_i\in C\cup D\cup \{s\}$, $(u_i,v_i)\in Y_s$. As a result,
  there is an edge $u_iv_i$ of weight $\dist_G(u_i,v_i)$ in $\xdet$.

  Suppose $|P_i|>h$. Split $P_i$ arbitrarily into subpaths $Q_1,\ldots,Q_l$
  such that $\lfloor h/2\rfloor \leq |Q_j|\leq h$ for all $j=1,\ldots,l$.
  Consider some $j$ and
  let $Q_j$ be an $x_j\to y_j$ path.
  We will prove that $\dist_{\xdet}(x_j,y_j)\leq |Q_j|=\dist_G^h(x_j,y_j)$.
  It will follow that
  \begin{equation*}
    \dist_{\xdet}(u_i,v_i)\leq \sum_j\dist_{\xdet}(x_i,y_j)\leq \sum_j\dist_G^h(x_j,y_j)=\sum_j\dist_G(x_j,y_j)=\dist_G(u_i,v_i),
  \end{equation*}
  as desired.

  If $(x_j,y_j)\in Y_s$, then there is an edge $x_jy_j$ of weight $\dist_G^h(x_j,y_j)=\dist_G(x_j,y_j)$
  in $\xdet$ which immediately implies $\dist_{\xdet}(x_j,y_j)\leq \dist_G^h(x_j,y_j)$.
  Otherwise, we have $(x_j,y_j)\notin Y_s$.
  Then by the definition of $P_i$ and $Y_s$, we obtain that $V(Q_j)\cap (C\cup D)=\emptyset$.
  This means that $Q_j\subseteq G-(C\cup D)\subseteq G_0-C$ and thus $\dist_{G_0-C}^h(x_j,y_j)\leq \dist_G(x_j,y_j)$
  as $Q_j$ is a shortest path in $G$.
  Moreover, since $(x_j,y_j)\notin Y_s$ and $\dist_{G_0-C}^h(x_j,y_j)<\infty$, $\pi_{x_j,y_j}\in \Pi\setminus \Pi_D$, i.e.,
  $\pi_{x_j,y_j}$ exists and does not contain a vertex from $D$.
  As a result, $\pi_{x_j,y_j}\subseteq G$, i.e., $\dist_{G}(x_j,y_j)\leq |\pi_{x_j,y_j}|$.
  But we also have $|\pi_{x_j,y_j}|\leq \dist_{G_0-C}^h(x_j,y_j)$ by Theorem~\ref{l:collection-thres}, so we obtain:
  \begin{equation*}
    \dist_G(x_j,y_j)\leq |\pi_{x_j,y_j}|\leq \dist_{G_0-C}^h(x_j,y_j)\leq \dist_G(x_j,y_j).
  \end{equation*}
  All the inequalities above have to be in fact equalities. This
  proves that $\pi_{x_j,y_j}$ is a shortest $x_j\to y_j$ path in $G$.
  Finally, recall that $|\pi_{x_j,y_j}|=\dist_G(x_j,y_j)\geq \lfloor h/2\rfloor$. Consequently, there exists a vertex ${z=\beta(\pi_{x_j,y_j})\in H}$
  on the path $\pi_{x_j,y_j}\in \Pi\setminus \Pi_D$.
  We conclude $\dist_G(x_j,y_j)=|\pi_{x_j,y_j}[x_j,z]|+|\pi_{x_j,y_j}[z,y_j]|$.
  By the definition of the multisets $\fromh{z}(x_j)$ and $\toh{z}(y_j)$,
  $|\pi_{x_j,y_j}[x_j,z]|\in \fromh{z}(x_j)$ and $|\pi_{x_j,y_j}[z,y_j]|\in \toh{z}(y_j)$.
  Consequently, $\wei_{\xdet}(x_j,z)\leq |\pi_{x_j,y_j}[x_j,z]|$
  and $\wei_{\xdet}(z,y_j)\leq |\pi_{x_j,y_j}[z,y_j]|$, and thus
  there is a path $x_j\to y_j$
  of weight $\wei_{\xdet}(x_j,z)+\wei_{\xdet}(z,y_j)\leq \dist_G(x_j,y_j)$ in $\xdet$.
\end{proof}

\paragraph{Setting the parameters.} Let us now try to set the parameters so that the maximum of the update and query time is minimized. 
The only term incurred by the query procedure that is not dominated by the terms of the
previously obtained update bound is $\Ot(n^2/h)$.
Moreover, it can be easily verified that it is best to set $\mu:=\rho$, where $\rho\approx 0.529$ is the unique value
such that $\omega(1,\rho,1)=1+\rho$.

This way, we obtain the following combined updated/query bound:
\begin{equation*}
  O\left(\frac{n^2h}{\Delta}+\frac{h^2n^{2+\rho}}{b}\cdot \left(1+\frac{n}{\tau \Delta}\right)+\frac{n^2b}{\Delta}+
  h^2n^{\omega(1,\nu,1)-\nu}+\frac{h^2n^2}{\tau}+h^3\tau n\cdot (\Delta+n^\nu)+\frac{h^2n^{2+\nu}}{\tau\Delta}+\frac{n^2}{h}\right).
\end{equation*}
By setting the parameters $h=n^{0.031}$, $\Delta=n^{0.78}$, $t=n^{0.093}$, $b=n^{0.749}$, $\nu=\rho$
obtained using the online balancing tool~\cite{Complexity},
we get an $O(n^{1.969})$ worst-case update and query bounds.

It is also interesting to see what is the exact dependence on $\omega$ we can get.
One can prove that for large enough $\Delta$ it is optimal to
set $\nu:=\rho$ so that $n^\nu<\Delta$, and also $b:=\Delta/h$, and $\tau:=h^3$.
Then, we get the bound
\begin{equation*}
  \Ot\left(\frac{hn^{2+\rho}}{\Delta}+\frac{n^{3+\rho}}{\Delta^2}+\Delta h^6n+\frac{n^2}{h}\right)
\end{equation*}
By balancing the two last terms, we get $\Delta=n/h^7$.
Then the first term ($h^8n^{1+\rho}$) becomes dominated by the second ($h^{14}n^{1+\rho}$).
By further balancing the latter with $n^2/h$ we get $h=n^{(1-\rho)/15}$.
The final bound becomes (by $\rho\leq \frac{1}{4-\omega}$):
\begin{equation*}
  \Ot\left(n^{2-\frac{1-\rho}{15}}\right)\subseteq \Ot\left(n^{2-\frac{3-\omega}{60-15\omega}}\right).
\end{equation*}
This bound is subquadratic assuming any truly subcubic
matrix multiplication algorithm is used and equals $\Ot(n^{2-1/30})$ if $\omega=2$ (then, $\rho=0.5$).

\paragraph{Deamortization.}
The heavy super-quadratic preprocessing invoking Theorem~\ref{t:collection}, initializing the set $Y$, and constructing the hitting set (which are the only source of amortization here) happens in a coordinated
way, once every $\Delta$ updates.
Consequently, the data structure can be deamortized using a standard technique
of running two copies of the data structure in parallel and spreading out the reset computation over several updates. While one copy of the algorithm is performing the reset, the other copy is used to answer the queries.
For details, we refer to e.g.,~\cite{AbrahamCK17, BrandNS19, Thorup05}.

We have thus proved:

\begin{theorem}\label{t:ss-distances}
  Let $G=(V,E)$ be unweighted digraph. There exists a deterministic data structure supporting
  fully dynamic single-edge updates to $G$ and computing exact distances from an arbitrary query source $s\in V$
  to all other vertices in $G$
  in $O(n^{1.969})$ worst-case time. 

  If $\omega=2$, the bound becomes $\Ot(n^{2-1/30})$.
\end{theorem}

\paragraph{SSSP tree reporting.} We now explain how the data structure of Theorem~\ref{t:ss-distances} can be extended to also
output a single-source shortest paths tree in subquadratic time, albeit at the cost of an additional polynomial overhead.

We apply Lemma~\ref{l:blocks} with $h:=n$ on top of the obtained data structure maintaining
single-source distances. In terms of Lemma~\ref{l:blocks}, we have $\qsource(n,n)=O(n^{1.969})$ by Theorem~\ref{t:ss-distances},
and we want to support queries of the form $S=\{s\}$, $Y=\{s\}\times V$,
for a query source $s\in V$.
As a result, for any $b\in [1,n]$, there exists a data structure supporting single-edge updates
to $G$ in $O(n^{2.969}/b)$ time and reporting predecessors of shortest $n$-bounded paths from any
source vertex $s$ in $O(n^{2.969}/b+nb)$ time.
Clearly, the predecessors yield an SSSP tree from $s$.
By setting $b=0.985$, we obtain a deterministic tree-reporting SSSP data structure
with $O(n^{1.985})$ worst-case update and query time.

\section{Fully dynamic path-reporting transitive closure and approximate APSP}\label{s:tc-reporting}
In this section, we consider the fully dynamic transitive closure and $(1+\eps)$-approximate
APSP problems (in unweighted digraphs)
under \emph{vertex updates} to~$G$.
Under the assumption that $\omega\geq 2.12$,
we describe a deterministic data structure with $O(n^{\omega-\gamma})$ worst-case update
time, for some $\gamma>0$ dependent on $\omega$, that explicitly maintains 
the transitive closure of a matrix of distance estimates.
In the case of approximate distances, the data structure also supports
reporting any $s\to t$ path of length at most $(1+\eps)\cdot \dist_G(s,t)$
in optimal $O(\dist_G(s,t))$ time.
For reachability, an even stronger reporting interface is achieved:
reachability trees in~$G$ from all possible sources $s\in V$
are explicitly stored after each update.
Clearly, this also allows for reporting paths in~$G$ between arbitrary pairs of vertices
in optimal time.

The overall approach is similar
as in the fully dynamic SSSP data structure from Section~\ref{s:sssp}.
Namely, we will operate in phases of length $\Delta$ and
use analogous global data structures shared among the phases.
Every phase starts by computing a pair $(C,\Pi)$ as given by Theorem~\ref{t:collection}.

\paragraph{The data structure~$\dstree$.}
Recall that the global data structures $\dsy$ and $\dstree$ in Section~\ref{s:sssp}, both implemented
as data structures of \Cref{t:h-bounded-submatrix-batch} (with possibly some extensions of Section~\ref{s:blocks}),
rely only on \emph{single-edge} updates (so that they reflect $G$ at all times) and vertex switch-offs.
Here, we will use the data structure $\dstree$ from Section~\ref{s:sssp}, required by Theorem~\ref{t:collection},
in a slightly different way.

Instead of passing the updates to $G$ to $\dstree$ immediately,
we will only do so at the end of a phase, that is, in a batch of at most $\Delta$ updates.
This is done with the vertex batch-update operation of \Cref{t:h-bounded-submatrix-batch}.
To see that this is sufficient,
recall that the only purpose of $\dstree$ 
was to construct the set $C$ and a collection of paths $\Pi$.
For this task, $\dstree$ only required supporting switching vertices
off interleaved with $h$-bounded-shortest paths queries from a single source.
We thus require $\dstree$ to reflect the current graph $G$ only once per $\Delta$ updates.
Since a batch update is performed on $\dstree$ only once per $\Delta=n^\delta$ updates,
this incurs an $\Ot\left(h^2 n^{\omega(1,\delta,1)+1}/(\Delta b)\right)$ worst-case cost per update.
As a result, the total per-update cost of initializing a phase
is now:
\begin{equation*}
  O\left(\frac{h^2 n^{\omega(1,\delta,1)+1}}{\Delta b}+\frac{h^2n^{2+\rho}}{b}\cdot\left(1+\frac{n}{\Delta\tau}\right)+\frac{n^2b}{\Delta}\right).
\end{equation*}
\newcommand{\dsyb}{\mathcal{B}}
\newcommand{\dsource}{\mathcal{S}}
\newcommand{\dpair}{\mathcal{P}}

\paragraph{The set $Y$.} We also redefine the set $Y\subseteq V\times V$, for a fixed phase, to only contain
the pairs $(s,t)$ such that $\pi_{s,t}\in \Pi_D$.
As a result, by Theorem~\ref{t:collection}, we have $|Y|=O(\Delta h\tau n)$ at all times.

\paragraph{The predecessor reporting data structure $\dsyb$.}
Let $\dsource$ and $\dpair$ be two data structures of Theorem~\ref{t:h-bounded-submatrix-batch}, for
the parameter $\mu$ to be set later, responsible for maintaining $h$-bounded distances
in $G$.
While the phase proceeds, these data structures will not accept edge/vertex updates.
Only at the end of the phase they are updated with a batch of $\Delta$ vertex updates
from the current phase.

The data structures $\dsource$ and $\dpair$ have different purposes.
$\dsource$ is initialized so that it can maintain $\dist_{G}^h(s,t)$ \emph{explicitly}
for all $n^2$ pairs $s,t\in V$ (i.e., $Y:=V\times V$ in terms of Theorem~\ref{t:h-bounded-submatrix-batch}),
whereas $\dpair$ is capable of answering arbitrary $\dist_{G}^h(s,t)$ queries
in $\Ot(h^2n^\mu)$ time.
This way, $\dsource$ allows very fast queries at the cost of $\Ot(h^2n^2)$ worst-case update time,
whereas $\dpair$ answers queries slower, but can handle updates
more efficiently, in $\Ot(h^2(n^{\omega(1,\mu,1)-\mu}+n^\mu))$ worst-case time.

We use data structures $\dsource$ and $\dpair$ to construct a predecessor reporting data structure $\dsyb$
of Lemma~\ref{l:blocks2} storing the graph $G-D$, for $b:=\beta$, where $\beta\in [1,n]$
is a parameter to be set later.
Specifically, we first apply Lemma~\ref{l:blocks2}, and then also Lemma~\ref{l:induced-subgraph}
on top of that to enable cheap insertions to $D$ (we model
$G-D$ as the induced subgraph $G[W]$ with $W:=V\setminus D$).
The general purpose of $\dsyb$ is to reconstruct $h$-bounded shortest paths in $G-D$.
Note that in terms of the quantities used in Lemma~\ref{l:blocks2}, we have
$\qsource(n,h)=O(n)$, $\usource(n,h)=\Ot(h^2n^2)$, $\qpair(n,h)=\Ot(h^2n^\mu)$ and $\upair(n,h)=\Ot(h^2(n^{\omega(1,\mu,1)-\mu}+n^{1+\mu}))$.
Therefore, by issuing a query with $S=V$ (which trivially guarantees $Y\subseteq S\times V$),
$\dsyb$ can compute a predecessor vertex $q_{s,t}$ on some shortest $h$-bounded $s\to t$ paths
for all $(s,t)\in Y$ such that $\dist_{G-D}^h(s,t)<\infty$ in time
\begin{equation*}
  \Ot(n^2+(h^2n^\mu+\beta)\cdot |Y|).
\end{equation*}
Adding new vertices to $D$ corresponds to removal of a vertex from the subset $W$ issued to the data structure $\dsyb$,
and thus costs
\begin{equation*}
  \Ot(h^2n^2+(n/\beta)\cdot h^2(n^{\omega(1,\mu,1)-\mu}+n^{1+\mu}))
\end{equation*}
time.
Similarly, the batch-vertex update to $\dsyb$ at the end of a phase
costs $\Ot\left(\frac{h^2n^{\omega(1,\delta,1)+1}}{\beta\Delta}\right)$ amortized time per update.

\paragraph{Computing all-pairs distance estimates upon update.}
After an update, we compute the respective predecessors $q_{s,t}$ for all $(s,t)\in Y$ using $\dsyb$.
We now introduce a certain collection of paths $\Pi'=(\pi')_{s,t\in V}$ derived from $\Pi$, $D$, and the computed predecessors.
For each $s,t\in V$, $\pi'_{s,t}$ is either
a path in $G-D$, or is left undefined (set to $\perp$).
Specifically, we set:
\begin{equation*}
  \pi'_{s,t}=\begin{cases}
    \pi_{s,t} & \text{ if }\pi_{s,t}\in \Pi\setminus \Pi_D,\\
    \pi'_{s,q_{s,t}} \cdot q_{s,t}t & \text{ if }\pi_{s,t}\in \Pi_D\text{ and }\dist_{G-D}^h(s,t)\neq\infty\text{ and }\pi'_{s,q_{s,t}}\neq\perp,\\
    \perp & \text{ otherwise. }
  \end{cases}
\end{equation*}
The definition is valid since $\pi_{s,t}\in \Pi_D$ and $\dist_{G-D}^h(s,t)\neq\infty$ implies
that $(s,t)\in Y$ and thus $\dist^h_{G-D}(s,q_{s,t})=\dist^h_{G-D}(s,t)-1$.
Note that the collection $\Pi'$ can be explicitly constructed from
the definition in $O(n^2h)$ time once the predecessors $q_{s,t}$ are computed.
For simplicity, define $|\perp|=\infty$.

\begin{lemma}\label{l:coll-prim}
  For any $s,t\in V$ such that no shortest $h$-bounded $s\to t$ path in $G$ goes
  through $C\cup D$, we have $|\pi_{s,t}'|=\dist_{G}^h(s,t)$.
\end{lemma}
\begin{proof}
  Recall that we denote by $G_0$ the graph at the beginning of the current phase.

  First consider the case $\dist_G^h(s,t)=\infty$. Then either $\pi_{s,t}\notin \Pi$ which implies $\pi'_{s,t}=\perp$,
  or $\dist_{G_0-C}^h(s,t)<\infty$ and thus $\pi_{s,t}\in \Pi_D$ by $\dist_{G-D}^h(s,t)=\infty$.
  However, in such a case we set $\pi'_{s,t}:=\perp$ as well.

  Now let us assume that $\dist_G^h(s,t)$ is finite. We prove the claim by induction
  on $\dist_G^h(s,t)$.
  Since no shortest $s\to t$ path goes through $C\cup D$, $\dist_G^h(s,t)=\dist_{G-D}^h(s,t)=\dist_{G-(C\cup D)}^h(s,t)\geq \dist_{G_0-C}^h(s,t)$.
  Hence, we obtain that $\pi_{s,t}\in \Pi$ exists and has length at most $\dist_{G_0-C}^h(s,t)\leq \dist_G^h(s,t)$.
  If $\pi_{s,t}\in \Pi\setminus \Pi_D$, then $\pi'_{s,t}=\pi_{s,t}\subseteq G-D\subseteq G$,
  so indeed we have $|\pi'_{s,t}|=\dist_G^h(s,t)$.
  Otherwise, $\pi_{s,t}\in \Pi_D$.
  By $\dist_{G-D}^h(s,t)<\infty$, $q_{s,t}$ exists and $\dist_G^h(s,q_{s,t})=\dist_G^h(s,t)-1$.
  Since no shortest $h$-bounded $s\to t$ path in~$G$ goes through $C\cup D$, neither
  does any shortest $h$-bounded $s\to q_{s,t}$ path in $G$.
  By the inductive assumption we have that $|\pi'_{s,q_{s,t}}|=\dist_G^h(s,q_{s,t})$.
  As a result, the path $\pi'_{s,t}=\pi'_{s,q_{s,t}}\cdot q_{s,t}t$ indeed
  has length $\dist_G^h(s,t)$.
\end{proof}

With the collection $\Pi'$ in hand, in $O(n^2h)$ time we can deterministically compute an $\Ot(n/h)$-sized hitting set $H\subseteq V$
of all $\pi'_{s,t}$ satisfying $|\pi_{s,t}'|=h$ using Lemma~\ref{l:hitting}.

\begin{lemma}\label{l:path-decomp}
  Let $s\in V\setminus (C\cup D\cup H)$ and $t\in V$ be such that $\dist_G(s,t)<\infty$. Some shortest $s\to t$ path in $G$
  either:
  \begin{enumerate}[(a)]
    \item is equal to $\pi'_{s,t}$, or
    \item has a prefix $\pi'_{s,q}\cdot qr$ for some $r\in C\cup D\cup H$, where $q=r$ or $qr\in E(G)$ and $|\pi'_{s,q}\cdot qr|\geq 1$.
  \end{enumerate}
\end{lemma}
\begin{proof}

  First suppose that $\dist_G(s,t)\leq h$, i.e., $\dist_G^h(s,t)\neq\infty$.

  If no shortest $s\to t$ path in $G$ goes through $C\cup D$, then by Lemma~\ref{l:coll-prim},
  $\pi'_{s,t}$ is a shortest path, i.e., item~(a) is satisfied.
  
  Let us thus assume that some shortest $s\to t$ path goes through $C\cup D$.
  There exists an $r\in C\cup D$ such that $\dist_G(s,t)=\dist_G(s,r)+\dist_G(r,t)$;
  if there are many such~$r$, pick one with minimal $\dist_G(s,r)$.
  By $s\notin C\cup D$, $\dist_G(s,r)\geq 1$.
  If $\dist_G(s,r)=1$, then some shortest $s\to t$ path starts with the edge $sr$,
  i.e., item~(b) is satisfied with $q=s$.
  Otherwise, if $\dist_G(s,r)\geq 2$,
  then by the choice of $r$, any $s\to r$ path has no vertices
  in $C\cup D$ except $r$. As a result, if $q\in V\setminus (C\cup D)$ is the penultimate
  vertex on one such $s\to r$ path, we obtain that no $s\to q$ path in~$G$ contains
  a vertex from $C\cup D$.
  By applying Lemma~\ref{l:coll-prim} to $(s,q)$, we obtain that $\pi'_{s,q}$
  is a shortest $s\to q$ path in $G$.
  Therefore, we conclude that some shortest $s\to t$ path in $G$
  has a prefix $\pi'_{s,q}\cdot qr$, where $r\in C\cup D$, as desired.

  Now consider the case $\dist_G(s,t)>h$.
  Let $x\in V$ be such that $\dist_G(s,t)=\dist_G(s,x)+\dist_G(x,t)$ and $\dist_G(s,x)=h$.
  We have already proved the lemma for the pair $(s,x)$.
  If there is a shortest $s\to x$ path satisfying (b), there is
  also clearly a shortest $s\to t$ path satisfying (b).
  It may also happen that $\pi'_{s,x}$ is a shortest $s\to x$ path in $G$.
  But then, there exists a vertex $u\in H$, $u\neq s$, on $\pi'_{s,x}$.
  By applying the lemma once again to the pair $(s,u)$ satisfying $\dist_G(s,u)\leq h$, we obtain
  that either (b) holds for $(s,u)$ and thus also for $(s,t)$,
  or (a) holds for $(s,u)$, which together with $s\notin H$ implies that $|\pi'_{s,u}|\geq 1$ and
  thus (b) holds for $(s,t)$ with $q=r=u\in H$.
\end{proof}
By proceeding symmetrically, one can prove a ``suffix'' version of Lemma~\ref{l:path-decomp}:
\begin{lemma}\label{l:path-decomp-suf}
  Let $s\in V$ and $t\in V\setminus (C\cup D\cup H)$ be such that $\dist_G(s,t)<\infty$. Some shortest $s\to t$ path in $G$
  either:
  \begin{enumerate}[(a)]
    \item is equal to $\pi'_{s,t}$, or
    \item has a suffix $rq\cdot \pi'_{q,t}$ for some $r\in C\cup D\cup H$, where $q=r$ or $rq\in E(G)$ and $|rq\cdot \pi'_{q,t}|\geq 1$.
  \end{enumerate}
\end{lemma}

Let us now recall the definition of the min-plus product $\star$. If $A$ is an $a\times b$ matrix
and $B$ is an $b\times c$ matrix, then $A\star B$ is a $a\times c$ matrix such that
$(A\star B)_{i,j}=\min_{k=1}^b\{A_{i,k}+B_{k,j}\}$.

Define a matrix $A\in  \{0,\ldots,h,\infty\}^{n\times n}$ as follows:
\begin{equation*}
  A_{u,v}=\begin{cases}
    0 &\text{ if }u=v\\
    1 &\text{ if }u\neq v\text{ and }uv\in E(G),\\
    |\pi'_{u,v}| &\text{ if }u\neq v\text{ and }uv\notin E(G)\text{ and }\pi'_{u,v}\neq\perp,\\
    \infty &\text{ otherwise. }
  \end{cases}
\end{equation*}

Let $A[S,T]$ denote the submatrix of $A$ with rows $S$ and columns $T$, where $S,T\subseteq V$.
Let $X:=C\cup D\cup H$.
Define the matrix $A^*$ as:
\begin{equation}\label{eq:matrix}
  A^*=\min\left(A, (A\star A[V,X]) \star \left(A[X,V]\star A\star A[V,X]\right)^n \star (A[X,V]\star A)\right),
\end{equation}
where the matrix powering is using the min-plus matrix product $\star$,
and the minimum is element-wise.

\begin{lemma}\label{l:tc}
  For any $s,t\in V$, $A^*_{s,t}=\dist_G(s,t)$.
\end{lemma}
\begin{proof}
  Note that $A_{u,v}\geq \dist_G(u,v)$ for all $u,v\in V$.
Moreover, the min-plus product of any two matrices with rows and columns in $V$ satisfying this condition
also satisfies this condition.
  This proves that $A^*_{u,v}\geq \dist_G(u,v)$.

  Let us now prove that $A^*_{s,t}\leq \dist_G(s,t)$.
  We first consider the case $s,t\in X$.
  We proceed by induction on $d:=\dist_G(s,t)$ and prove that if
  $B:=A[X,V]\star A\star A[V,X]$, then $(B^d)_{s,t}\leq \dist_G(s,t)$ (where the powering is using the min-plus product $\star$).
  Let us first note that since $A$ has zeros on the diagonal,
  $(B^{i+1})_{s,t}\leq (B^i)_{s,t}$ for any $i\geq 0$ and $B_{s,t}\leq A_{s,t}$.

  The identity matrix $I$ for the min-plus product has zeros on the diagonal and $\infty$ everywhere else.
  If $d=0$, then $s=t$ and $B^0=I$, so indeed $(B^d)_{s,t}=I_{s,s}=0$.

  Suppose $d\geq 1$.
  Let $u\in V$ be any vertex such that $su\in E(G)$ and $\dist_G(u,t)=d-1$.
  We have $A_{s,u}=1$ and thus also $B_{s,u}\leq 1$.
  If $u\in X$, then both $B_{s,u}\leq 1$ and $(B^{d-1})_{u,t}\leq d-1$ hold by the inductive
  assumption, and thus indeed $(B^d)_{s,t}\leq d$.

  Now assume $u\notin X$.
  By Lemma~\ref{l:path-decomp} applied to $(u,t)$, we have two cases.
  If case (a) applies, $\pi'_{u,t}\neq\perp$, so $B_{u,t}\leq A_{u,t}\leq |\pi'_{u,t}|=\dist_G(u,t)\leq d-1$ and thus again we obtain $(B^{d-1})_{u,t}\leq d-1$
  and consequently $(B^d)_{s,t}\leq d$.
  If case (b) applies, then for some $q,r$, where $q\in V$,
  $r\in X$ and $\dist_G(r,t)\leq d-1$, we have
  $\dist_G(s,r)=|su\cdot \pi'_{u,q}\cdot qr|$,
  where possibly $q=r$ or $s=q$.
  But this implies 
  $B_{s,r}\leq (A[X,V])_{s,u}+A_{u,q}+(A[V,X])_{q,r})\leq \dist_G(s,r)$.
  By the inductive assumption, $(B^{d-1})_{r,t}\leq \dist_G(r,t)$, so $(B^d)_{s,t}\leq \dist_G(s,r)+\dist_G(r,t)=\dist_G(s,t)$ as well.

  We have therefore proved that for any $s,t\in X$, $\dist_G(s,t)<\infty$
  implies that $(B^n)_{s,t}\leq \dist_G(s,t)$. This in turn implies $A^*_{s,t}\leq \dist_G(s,t)$ for such pairs $(s,t)$.
  
  If $s\notin X$ and $t\in X$, then by Lemma~\ref{l:path-decomp}, either (a) $\pi'_{s,t}$ is a shortest
  $s\to t$ path, which implies $A_{s,t}\leq \dist_G(s,t)$, or (b) for some $q\in V$ and $r\in X$ with $\dist_G(r,t)<\dist_G(s,t)$,
  there is a shortest $s\to t$ path starting with $\pi'_{s,q}\cdot qr$,
  which implies $\dist_G(s,t)=\dist_G(s,r)+\dist_G(r,t)$ and $A_{s,q}+A_{q,r}\leq \dist_G(s,r)$.
  In both cases we conclude $\left((A\star A[V,X])\star B^n\right)_{s,t}\leq \dist_G(s,t)$.
  By Lemma~\ref{l:path-decomp-suf} we can symmetrically prove
  that if $s\in X$ and $t\notin X$, then $(B^n\star (A[X,V]\star A))_{s,t}\leq \dist_G(s,t)$.
  
  Finally, if $s\notin X$ and $t\notin Y$, then we analogously
  argue that either $A_{s,t}\leq \dist_G(s,t)$ or
  for some $r\in X$ we have
  $\dist_G(s,t)=\dist_G(s,r)+\dist_G(r,t)$, $(A\star A[V,X])_{s,r}\leq \dist_G(s,r)$ and 
  $(B^n\star (A[X,V]\star A))_{r,t}\leq \dist_G(r,t)$, which implies
  that  $\left((A\star A[V,X])\star B^n\star (A[X,V]\star A)\right)_{s,t}\leq \dist_G(s,t)$.
  Therefore, we conclude that for the general case ${s,t\in V}$,
  we have either $A_{s,t}\leq \dist_G(s,t)$ or
  \linebreak
  $\left((A\star A[V,X])\star B^n\star (A[X,V]\star A)\right)_{s,t}\leq \dist_G(s,t)$,
  as desired.
\end{proof}

Lemma~\ref{l:tc} yields a way of computing exact all-pairs distances in $G$ based on the collection $\Pi'$
provided
that an efficient algorithm computing rectangular min-plus products is available.
However, the known min-plus matrix product algorithms are not efficient enough
to give an improved fully dynamic exact APSP algorithm this way.
Nevertheless, we will argue that Lemma~\ref{l:tc} is enough for deterministically computing the matrix
of $(1+\eps)$-approximate distances in $G$ polynomially faster than in $O(n^\omega)$ time, for any $\omega>2$.
To this end, we will need the following $(1+\eps)$-approximate min-plus matrix product $\star_\eps$ of Zwick~\cite[Section 8]{Zwick02}:
\begin{theorem}\label{t:appr-dist-prod}{\upshape{\cite{Zwick02}}}
  Let $A,B$ be two $n\times n$ matrices with elements from $\{0,\infty\}\cup [1,N]$. Let $\eps>0$.
  Then, in $\Ot(n^\omega\log(N)/\eps)$ deterministic time one can compute a matrix $A\star_\eps B$ such that
  for all $i,j\in [n]$:
  \begin{equation*}
    (A\star B)_{i,j}\leq (A\star_\eps B)_{i,j}\leq (1+\eps)\cdot (A\star B)_{i,j}.
  \end{equation*}
\end{theorem}
\begin{remark}\label{r:appr-dist-prod}{\upshape{\cite{Zwick02}}}
  Theorem~\ref{t:appr-dist-prod} can be easily generalized to the rectangular case, i.e., if $A$ and~$B$ are
  $n^a\times n^b$ and $n^b\times n^c$ matrices, then the approximate min-plus product can be computed
  in $\Ot(n^{\omega(a,b,c)}\log(N)/\eps)$ time. Moreover, the approximate min-plus product can be extended
  to deterministically produce witnesses within a time bound that is a $\polylog(n)$ factor away.
\end{remark}

Let $\eps'>0$.
Suppose we evaluate the formula~\eqref{eq:matrix} by replacing the min-plus product $\star$
with $(1+\eps')$-approximate min-plus product $\star_{\eps'}$ of Theorem~\ref{t:appr-dist-prod} (and Remark~\ref{r:appr-dist-prod})
and using the folklore repeated-squaring algorithm for matrix powering.
Then, only $O(\log{n})$ matrix products will be performed and
each of the products involved 
will take $\Ot(n^{\omega(1,\alpha,1)}\log(N)/\eps')$ time,
where $\alpha\in [0,1]$ is such that $|C\cup D\cup H|\leq n^{\alpha}$, and
$N$ is a bound on the largest finite number arising in this computation.
Indeed, in all the matrix products involved, one of the matrix dimensions
dimensions is bounded by $n^\alpha$ (we will later set $\tau,h,\Delta$ so that $\alpha<1$).
Moreover, since there are $O(\log{n})$ matrix products, the approximation error
at the end is bounded by $(1+\eps')^{O(\log{n})}$.
This also implies the bound $N\leq n\cdot (1+\eps')^{O(\log{n})}$.
By using a sufficiently small $\eps'=\Theta(\eps/\log{n})$, we can guarantee
that the entire computation is $(1+\eps)$-approximate and $N=\Ot(n)$.
We conclude that the matrix of $(1+\eps)$-approximate distance estimates
can be computed based on $\Pi'$ in $\Ot(n^{\omega(1,\alpha,1)}/\eps)$ time.

It is worth noting that
if we wanted to compute the transitive closure of $G$ (which can be seen as any finite approximation
of APSP) instead, we could
use the standard boolean matrix product~$\cdot$ instead of the approximate min-plus product $\star_{\eps'}$ here.

\paragraph{Approximate shortest path reporting.}
We now turn to explaining how the data structure can be extended to enable efficient reporting of approximately shortest $s,t$-paths
while computing the matrix $A^*$ approximately. Recall that, by Remark~\ref{r:appr-dist-prod},
the approximate min-plus product of Theorem~\ref{t:appr-dist-prod} can also produce witnesses, i.e., such indices $w_{i,j}$
that
\begin{equation*}
  A_{i,w_{i,j}}+B_{w_{i,j},j}\leq (1+\eps')\cdot (A\star B)_{i,j}.
\end{equation*}
Suppose we record the witnesses for all the $\ell=O(\log{n})$ approximate
products involved when computing $A^*$ approximately using Equation~\ref{eq:matrix}.
More concretely, let the subsequent matrices obtained be $L_1,\ldots,L_\ell$,
where $\min(A,L_\ell)$ is a $(1+\eps)$-approximation of $A^*$. By Lemma~\ref{l:tc},
we have
\begin{equation*}
  \min(A_{s,t},(L_\ell)_{s,t})\leq (1+\eps)\dist_G(s,t).
\end{equation*}
Put $L_0:=A$ and a define paths $P_{0,s,t}$ corresponding to the entries of $A$:
\begin{equation*}
  P_{0,s,t}=\begin{cases}
    \emptyset &\text{if }s=t,\\
    st &\text{if }st\in E(G),\\
    \pi'_{st} &\text{if }s\neq t, st\notin E(G),\text{and }A_{s,t}\neq \infty,\\
    \perp &\text{otherwise.}
  \end{cases}
\end{equation*}
Put $|\perp|=\infty$. When computing the matrix $L_k$ as a product of a submatrix of $L_i$ and a submatrix of $L_j$ (along with witnesses $w_{\cdot,\cdot}$),
where $i,j<k$, for all $s,t$ such that $(L_k)_{s,t}$ is defined, we define the path $P_{k,s,t}$ as follows:
\begin{equation*}
  P_{k,s,t}=\begin{cases}
    P_{i,s,w_{s,t}}\cdot P_{j,w_{s,t},t} & \text{if }w_{s,t}\notin \{s,t\},\\
    P_{i,s,t} & \text{if }w_{s,t}=t,\\
    P_{j,s,t} &\text{if }w_{s,t}=s.
  \end{cases}
\end{equation*}
In the implementation, we do not explicitly copy or concatenate the paths define above, but rather
store the pointers to the previously represented paths. This costs only $\Ot(n^2)$ additional time.

It is easy to prove inductively that for all $k,s,t$ such that $(L_k)_{s,t}\neq\infty$,
$P_{k,s,t}$ is an $s\to t$ path in $G$ of length at most $(L_k)_{s,t}$.
Consequently, for all $s,t$ such that $\dist_G(s,t)\neq\infty$, we have $\min(|P_{0,s,t}|,|P_{\ell,s,t}|)\leq (1+\eps)\dist_G(s,t)$. 
As a result, using the stored pointers we can report an $(1+\eps)$-approximately
shortest $s\to t$ path for any $(s,t)\in V$ in optimal
$O(\dist_G(s,t))$ time.

\paragraph{Single-source reachability trees from all sources.}
While in the case of approximately shortest paths and a fixed source $s\in V$, we can report
an $s\to t$ path of length $(1+\eps)\dist_G(s,t)$ for any $t\in V$ efficiently, it is not
clear how to organize such paths into a tree of approximately shortest paths
rooted at $s$.
However, as we show next, this is possible in the case of transitive closure.
Specifically, 
we now show that one can
compute \emph{reachability} trees from all
possible sources $s\in V$ in $G$ using the approximate matrix $A^*$,
if any constant $\eps\in (0,1)$ is used.

We first compute in $O(n^2)$ time the strongly connected components (SCCs) of $G$,
using any classical linear-time algorithm, e.g., that of~\cite{Tarjan72}.
\newcommand{\gscc}{G_{\mathrm{SCC}}}

\begin{lemma}\label{l:gscc}
  A subgraph $\gscc\subseteq G$ with $O(n)$ edges such that $\gscc$ preserves the strongly
  connected components of $G$ can be computed in $O(n^2)$ time.
\end{lemma}
\begin{proof}
  For any strongly connected component $S\subseteq V$ of $G$, include in $\gscc$
  a pair of trees in $G[S]$: a reachability tree from and to any $s\in S$ in the subgraph $G[S]$.
  These trees can be computed via graph search in time linear in the number of edges of $G[S]$.
  Since the subgraphs $G[S]$ are clearly disjoint for distinct $S$,
  computing $\gscc$ takes $O(n^2)$ time.
\end{proof}
Now let us define a matrix $A'$ such that $A'_{s,t}=0$ if
the vertices $s$ and $t$ are \emph{not} strongly connected in $G$
and either $A_{s,t}\neq \infty$ or $(A[V,X]\star_\eps A[X,V])_{s,t}\neq\infty$ holds,
and $A'_{s,t}=\infty$ otherwise.

Let us compute the matrix product
$F_V:=A^*[V,X]\star_\eps A'[X,V]$ along with the matrix $W_V$ of witnesses, that is,
for every $u,v$ such that $(F_V)_{u,v}\neq\infty$, $(W_V)_{u,v}$ is equal to some $z\in X$
such that $A^*_{u,z}\neq\infty$ and $A'_{z,v}\neq\infty$.
Similarly, we compute the matrix $W_X$ of witnesses of
the product $F_X:=A^*[V,V]\star_\eps A'[V,X]$.
Finally, let us compute the matrix $L$ of the witnesses of the product $A[V,X]\star_\eps A[X,V]$.
Again, in all the additional matrix products here, one of the dimension is bounded by $|X|\leq n^\alpha$,
so computing the required information takes $\Ot(n^{\omega(1,\alpha,1)}/\eps)$ time.

\begin{lemma}\label{l:reach-trees}
  For any $a,b\in V$, define $\pi^*_{a,b}$ to be either the edge 
  $ab$ if $ab\in E(G)$ or the path $\pi'_{a,b}$ if $\pi'_{a,b}\neq\perp$ (if both exist, pick an arbitrary one).
  If none of these exist, set $\pi^*_{a,b}:=\perp$.

    Let $s\in V$. Let $G_s\subseteq G$ be a subgraph containing:
  \begin{itemize}
    \item the path $\pi^*_{s,t}$ for all $t\in V$ such that $\pi^*_{s,t}\neq\perp$,
    \item for each $t\in V$ such that $(F_V)_{s,t}=1$, 
      the paths $\pi^*_{x,w}$ and $\pi^*_{w,t}$, where
      where $x=(W_V)_{s,t}$ and $w=L_{x,t}$.
    \item for each $x\in X$ such that $(F_X)_{s,x}=1$, the path $\pi^*_{w,x}$, where $w=W_{s,x}$.
    \item the subgraph $\gscc$.
  \end{itemize}
  Then a vertex is reachable from $s$ in $G$ if and only if it is reachable from $s$ in $G_s$.
\end{lemma}
\begin{proof}
  The ``$\impliedby$'' direction is trivial since $G_s$ is a subgraph of $G$.
  
  We will prove the ``$\implies$'' part by showing 
  that each SCC $S$ of $G$ that $s$ can reach is also reachable
  from $s$ in $G_s$.
  Since the SCCs of $G$ are preserved in $G$ by the inclusion of $\gscc$,
  this is clearly true for the SCC $S_s$ containing the vertex $s$.
  Since the graph obtained from $G_s$ by contracting the SCCs is a DAG,
  it is enough to argue that every other SCC $S'\neq S_s$ has in $G_s$
  an edge coming from another SCC $S''\neq S'$ reachable from $s$.

  Let $S\neq S_s$ be an SCC of $G$ such that $d=\min_{v\in S}\dist_G(s,v)\geq 1$.
  Let $t\in S$ be such that $\dist_G(s,t)=d$.
  Since every vertex $z$ strongly connected
  with $t$ has $\dist_G(s,z)\geq \dist_G(s,t)=d$,
  every vertex $y\neq t$ satisfying $\dist_G(s,t)=\dist_G(s,y)+\dist_G(y,t)$
  is not strongly connected with $t$, i.e., $y\notin S$.

  If $t\in X$, then let $jt$ be an arbitrary edge lying on
  some shortest $s\to t$ path in $G$.
  We have $\dist_G(s,j)=\dist_G(s,t)-1$, $A^*[s,j]<\infty$ by Lemma~\ref{l:tc},
  and $A'[j,t]<\infty$ since $j\notin S$.
  As a result, $(F_X)_{s,t}\leq A^*[s,j]+ A'[j,t]<\infty$.
  We conclude that there exists $g=(W_X)_{s,t}$ such that
  $A^*[s,g],A'[g,t]<\infty$. It follows that $g$ is not strongly connected
  with $t$ and $g$ is reachable from~$s$.
  By the definition of $G_s$, the $g\to t$ path $\pi^*_{g,t}$ (which exists by $A'[g,t]<\infty$)
  is included in $G_s$.
  This certifies that the SCC $S$ has an incoming
  edge from another SCC reachable from $s$.
  
  Finally, suppose $t\in V\setminus X$. 
  If $st\in E(G)$ or $\pi'_{s,t}\neq \perp$, then the corresponding edge or path $\pi^*_{s,t}$ is included in $G_s$.
  Since that edge/path originates in $S_s$ and ends in a different SCC
  $S$, it has to include at least one edge coming to $S$ from another SCC reachable from $s$.
  So suppose $\pi^*_{s,t}=\perp$.
  By Lemma~\ref{l:path-decomp-suf}, we obtain that
  some shortest $s\to t$ path has a suffix $rq\cdot \pi'_{q,t}$, where
  $r\in X$ and ${\dist_G(s,r)<d}$. By the definition of $t$, $r$ is in a different
  SCC than $t$.
  Hence, we have $A'[r,t]\leq A[r,q]+A[q,t]<\infty$.
  As a result, since $r$ is reachable from $s$, we have $(F_V)_{s,t}\leq A^*[s,r]+A'[r,t]<\infty$.
  Thus, there exists $x=(W_V)_{s,t}\in X$
  such that $A^*[s,x]<\infty$ and $A'[x,t]<\infty$ which in turn implies that $x$ and $t$ are not strongly connected
  and $x$ is reachable from $s$.
  Moreover, $w=L_{x,t}$ exists such that $A[x,w]<\infty$ and $A[w,t]<\infty$.
  Hence, there exist paths $\pi^*_{x,w}$ and $\pi^*_{w,t}$ in $G$.
  Their concatenation $\pi^*_{x,w}\cdot \pi^*_{w,t}$, included
  in $G_s$, certifies that $S$ has an incoming edge from another SCC reachable from~$s$.
 \end{proof}

 \begin{corollary}
   Given the matrices $F_X,W_X,F_V,W_V,L$ and the graph $\gscc$, one can compute reachability trees from
   all sources $s\in V$ in $O(n^2h)$ additional time.
  \end{corollary}
  \begin{proof}
    Note that for any $s\in V$, the subgraph $G_s\subseteq G$ preserving reachability from $s$
    has $O(nh)$ edges since it consists of the $O(n)$-sized subgraph $\gscc$ and
    $O(n)$ $\leq h$-hop paths $\pi^*_{a,b}$. Thus, some reachability tree in $G_s$ (and thus in $G$) can be computed in $O(nh)$ time.
  \end{proof}

\paragraph{Update time analysis.}
\newcommand{\mm}{\mathrm{MM}}

Let us now summarize the update time of the data structure. 
Recall that we set $\alpha\in [0,1]$ to be the smallest number
such that $\max(|C|,|H|,|D|)=O(\max(n/\tau,n/h,\Delta))\leq O(n^\alpha)$, and $\Delta=n^\delta$.
The amortized update time can be bounded by:
\begin{equation*}
  \begin{split}
  \Ot\left(\frac{h^2 n^{\omega(1,\delta,1)+1}}{\Delta b}+\frac{h^2n^{2+\rho}}{b}\cdot\left(1+\frac{n}{\Delta\tau}\right)+\frac{n^2b}{\Delta}+\Delta h\tau n(h^2n^{\mu}+\beta)+h^2n^2\right.\\
    \left.+\frac{h^2n(n^{\omega(1,\mu,1)-\mu}+n^{1+\mu})}{\beta}+\frac{h^2n^{\omega(1,\delta,1)+1}}{\beta\Delta}+n^{\omega(1,\alpha,1)}/\eps\right).
  \end{split}
\end{equation*}
Using the online balancing tool~\cite{Complexity}\footnote{\url{https://jvdbrand.com/complexity/?terms=2h\%2Bomega(1\%2Cdelta\%2C1)\%2B1-delta-b\%0A2h\%2B2\%2B0.529-b\%0A2h\%2B2\%2B0.529-b\%2B1-delta-tau\%0A2\%2Bb-delta\%0Adelta\%2Bh\%2Btau\%2B1\%2B2h\%2Bmu\%0Adelta\%2Bh\%2Btau\%2B1\%2Bbeta\%0A2h\%2B2\%0A2h\%2B1\%2Bomega(1\%2Cmu\%2C1)-mu-beta\%0A2h\%2B1\%2B1\%2Bmu-beta\%0A2h\%2Bomega(1\%2Cdelta\%2C1)\%2B1-delta-beta\%0Aomega(1\%2C1-tau\%2C1)\%0Aomega(1\%2C1-h\%2C1)\%0Aomega(1\%2Cdelta\%2C1)&a=1}}, we find that by setting
$\mu=0.355$, $\beta=n^{0.562}$, $b=n^{0.816}$, $\Delta=n^{0.523}$, $h=n^{0.1035}$ and $t=n^{0.104}$,
the obtained bound is $O(n^{2.293}/\eps)$.

It is also interesting to investigate for what values $\omega$ the obtained bound is better than $n^\omega$.
For this, let us again use the bound $\omega(1,a,1)\leq 2+a\cdot (\omega-2)$ (which yields $\rho\leq 1/(4-\omega)$)
to remove the dependency on rectangular matrix multiplication.

Since the only term that is non-increasing in $h$ is $n^{\omega(1,\alpha,1)}\leq n^{2+\alpha(\omega-2)}$,
having arbitrary $h=\poly{n}$ makes this term sub-$n^\omega$.
If all other terms with the $h$-terms skipped are also polynomially smaller $n^\omega$, then sufficiently
small polynomial $h$ makes all terms smaller than $n^\omega$.

Via experiments one can obtain that $\tau=h$ and $\beta=h^2n^\mu$ is a good choice,
and that the terms involving $\omega(1,\delta,1)$ are not bottlenecks if the parameters are chosen optimally.
As a result, to find the threshold for $\omega$, we need to pick $\mu,\Delta,b$ such
that the terms:
\begin{equation*}
  \frac{n^2b}{\Delta}, \frac{n^{3+\frac{1}{4-\omega}}}{b\Delta},\Delta n^{1+\mu},n^{3+\mu(\omega-4)}
\end{equation*}
are all polynomially smaller than $n^\omega$.
Observe that it is beneficial to balance the first two terms, which leads to $b=n^{\frac{5-\omega}{8-2\omega}}$.
Then, balancing with the third term yields $\Delta=n^{\frac{13-3\omega}{16-4\omega}-\mu/2}$.
Finally, balancing that with the fourth term yields $\mu=\frac{38-10\omega}{(16-4\omega)(9-2\omega)}$.
The final running time for $\omega>2$ after simplifications and taking into account small-polynomial $h,\tau=n^{O(\gamma)}$ becomes:
\begin{equation*}
  O\left(n^{2+\frac{\omega-1}{18-4\omega}+\gamma}+n^{\omega-\gamma}/\eps\right),
\end{equation*}
for any sufficiently small $\gamma>0$.
This bound is polynomially better than $n^\omega$ if $\omega>2+\frac{9-\sqrt{65}}{8}\approx 2.12$.

Amortization can be removed in a standard way, as was done in Section~\ref{s:sssp}.
The following theorem summarizes the data structure developed in this section.
\tcreporting*

\section{Deterministic fully dynamic transitive closure faster than $O(n^\omega)$ for all $\omega>2$}\label{s:tc-det}
In this section we show a simple deterministic data structure maintaining
the transitive closure of a digraph under vertex updates
polynomially faster than within $n^\omega$ worst-case update time, unless $\omega=2$.
Compared to the data structure of Theorem~\ref{t:tc-reporting}, we do not support
optimal-time path reporting here. However, as opposed to the path-reporting data structure,
we obtain a non-trivial worst-case update time for the entire range of possible values of $\omega>2$.

The bottleneck in the data structure of Theorem~\ref{t:tc-reporting} lies in recomputing
most of the pairwise $h$-bounded shortest paths.
These paths have also been used to compute
the hitting set $H$ deterministically. The set $H$ hit all the $h$-bounded
shortest paths in the graph $G-(C\cup D)$.
For transitive closure, we do not need to hit the shortest paths;
hitting a $h$-hop prefix of any path would also be enough.
Computing an $\Ot(n/h)$-sized hitting set achieving that deterministically
is still challenging in sub-$n^\omega$ time, though.
To deal with this problem, we will use slightly larger (but still sublinear)
hitting sets that are, at the same time, slightly less costly to recompute.
Using weaker (in terms of quality and size) hitting sets is the main idea behind
the data structure developed in this section.

Let $h,d\in [1,n]$ be parameters such that $d\cdot h\geq 6n$, to be chosen later.
First of all, we set up the data structure of Theorem~\ref{t:h-bounded-submatrix-batch} that
explicitly maintains the values $\dist^h_G(s,t)$ for all pairs $(s,t)\in V$ under
vertex updates to $G$ in $\Ot(h^2n^2)$ worst-case time per update.

The below lemma gives an algorithm to compute hitting sets
of non-necessarily shortest paths.

\begin{lemma}\label{l:weak-hitting-set}
  Let $G=(V,E)$ be any digraph. Then, for any $\ell\in [1,n]$, in $\Ot(n^\omega)$ time one can deterministically compute
  a hitting set $H\subseteq V$ of size $\Ot(n/\ell)$ satisfying the following.
  For any $s,t\in V$ with $\dist_G(s,t)<\infty$, either $\dist_G(s,t)< \ell$ or there is such a vertex $v\in H$
  that $\dist_G(s,v)\leq \ell$ and $\dist_G(v,t)<\infty$.
  If $G$ is strongly connected, then the running time can be improved to $\Ot(n^2)$.
\end{lemma}
\begin{proof}
  Suppose for any $s\in V$ we have computed a reachability tree $T_s$ from $s$ in $G$.
  Let $T_{s,h}$ be the tree $T_s$ pruned of all the vertices at depth more than $h$.
  Note that for all $t\in V(T_{s,h})$ at depth less than $h$, we necessarily have $\dist_G(s,t)<h$.
  
  For all other $t\in V$ either $\dist_G(s,t)=\infty$, or there exists
  a leaf $l_t$ at depth $h$ in $T_{s,h}$ such that one can reach $t$ from $s$ through $l_t$.
  As a result, if $H\subseteq V$ hits all the $h$-hop root-leaf paths in $T_{s,h}$,
  then there exists $v\in H$ such that $\dist_G(s,v)\leq h$
  and there exists a path $s\to v\to l_t\to t$.
  Therefore, it is enough to compute an $\Ot(n/h)$-sized hitting set $H$ of all
  the $h$-hop root-leaf paths in all of the trees $T_{s,h}$, $s\in V$.
  King~\cite[Lemma~5.2]{King99} gave a variant of the deterministic algorithm behind Lemma~\ref{l:hitting}
  that accomplishes precisely this task in $\Ot(n^2)$ time.

  Finally, the reachability trees from all the sources in a general graph
  can be computed deterministically in $\Ot(n^\omega)$ time~\cite{AlonGMN92}.
  If $G$ is strongly connected, then the same can be achieved in $O(n^2)$ time
    by running a graph search from every $s\in V$ in the sparse subgraph $\gscc\subseteq G$ from Lemma~\ref{l:gscc}.
\end{proof}

\newcommand{\blocks}{\mathcal{B}}

Lemma~\ref{l:weak-hitting-set} alone is not enough to deterministically
break through the $O(n^\omega)$ worst-case update bound.
To achieve that, proceed as follows upon each update.
First, compute in $O(n^2)$ the strongly connected components of $G$,
and their topological order using any classical algorithm.
Let $S_1,\ldots,S_k$ be the SCCs of $G$ in the topological order,
that is, any edge $uv$ such that $u\in S_i$ and $v\in S_j$,
satisfies $i\leq j$.
Next, partition the sequence $S_1,\ldots,S_k$ into at most $2n/d$ \emph{blocks} $\blocks$
of consecutive SCCs such that each block $B\in\blocks$ spanning $S_a\cup \ldots\cup S_b$ satisfies
either $a=b$ or $|S_a|+\ldots+|S_b|\leq d$.
Such a partition into blocks can be computed by repeatedly
locating the longest prefix
$S_i,\ldots,S_j$ of the remaining sequence $S_i,\ldots,S_k$
such that $|S_i|+\ldots+|S_j|\leq d$ (if it exists), and forming either a block $S_i\cup \ldots \cup S_j$ if $|S_i|\leq d$,
or a single-SCC block $S_i$ if $|S_i|>d$.
Note that every two consecutive blocks formed span more than $d$ vertices and this is why
the total number of blocks is at most $2n/d$.

We say that a block $S_a\cup \ldots\cup S_b=B\in\blocks$ is \emph{small} if $|S_a|+\ldots+|S_b|\leq d$.
Otherwise, we call the block \emph{large}. Recall that a large block necessarily consists
of a single SCC of $G$.

Next, for each block $B\in\blocks$ we compute the hitting set $H_B$ of Lemma~\ref{l:weak-hitting-set}
applied to the subgraph $G[B]$ for $\ell:=\lfloor \frac{dh}{6n}\rfloor$.
Note that we have $|H_B|=\Ot\left(\frac{|B|n}{dh}\right)$.
Set $H^*:=\bigcup_{B\in\blocks}H_B$.
We obtain that $|H^*|=\Ot\left(\sum_{B\in \blocks}\frac{|B|n}{dh}\right)=\Ot\left(\frac{n}{dh}\cdot \sum_{B\in\blocks}|B|\right)=\Ot\left(\frac{n^2}{dh}\right)$.

Observe that the total time spent on computing hitting sets in small blocks
can be bounded by $\Ot((n/d)\cdot d^\omega)=\Ot(nd^{\omega-1})$ since $|\blocks|=O(n/d)$.
Every large block is strongly connected and thus the hitting sets
of large blocks can be computed in $\Ot\left(\sum_{B\in\blocks}|B|^2\right)=\Ot(n^2)$ time.

\begin{lemma}\label{l:block-hitting}
  Let $s,t\in V$ be such that $\dist_G(s,t)<\infty$. Then either $\dist_G(s,t)\leq h$ or there exists $v\in H^*$
  such that $\dist_G(s,v)\leq h$ and $\dist_G(v,t)<\infty$.
\end{lemma}
\begin{proof}
  Suppose $\dist_G(s,t)>h$. Let $p\in V$ be any vertex such that $\dist_G(s,p)=h$ and $\dist_G(s,t)=\dist_G(s,p)+\dist_G(p,t)$.
  Note that any shortest $s\to p$ path $P$ in $G$ can be expressed as $P_0\cdot e_1\cdot P_1\cdot \ldots \cdot e_q\cdot P_q$,
  where each $P_i$ is a (possibly $0$-hop) maximal subpath containing vertices of a single block $B(P_i)\in \blocks$,
  and $e_i$ is an edge between distinct blocks $B(P_{i-1})$ and $B(P_i)$.
  Recall that the partition of SCCs into blocks respects the topological ordering
  so each $B(P_{i-1})$ is strictly earlier than $B(P_i)$ in the ordering.
  It follows that all the blocks $B(P_i)$ are distinct and thus $q\leq |\blocks|\leq 2n/d$.
  As a result, we obtain that some path $P_j$ has at least
  \begin{equation*}
    \frac{\sum_{i=0}^q|P_i|}{q+1}=\frac{|P|-q}{q+1}=\frac{h-q}{q+1}\geq \frac{h}{q+1}-1\geq \frac{h}{2n/d+1}-1\geq \frac{dh}{3n}-1\geq \frac{dh}{3n}-\frac{dh}{6n}=\frac{dh}{6n}\geq \ell
  \end{equation*}
  edges. Let $x,y$ be the endpoints of $P_j$, i.e., $P_j$ is an $x\to y$ path. 

  Note that $|P_j|=\dist_{G[B(P_j)]}(x,y)=\dist_G(x,y)\geq \ell$
  since $P$ is a shortest path in $G$.
  As a result, by Lemma~\ref{l:weak-hitting-set} we obtain that
  there exists some vertex $v\in H_{B(P_j)}\subseteq H^*$
  such that $\dist_{G[B(P_j)]}(x,v)\leq \ell$ and $\dist_{G[B(P_j)]}(v,y)<\infty$.
  By that, we have:
  \begin{align*}
    \dist_G(s,v)&\leq \dist_G(s,x)+\dist_G(x,v)\\
                &= \dist_G(s,x)+\dist_{G[B(P_j)]}(x,v)\\
                &\leq  \dist_G(s,x)+\ell\\
                &\leq |P_0\cdot e_1\cdot P_1\cdot \ldots e_{j-1}\cdot P_{j-1}\cdot e_j|+|P_j|\\
                &\leq |P|=h.\\
  \end{align*}
  Above we have also used a fact that for two vertices inside a single block,
  every path between them needs to be fully contained in the block.
  We also have:
  \begin{align*}
    \dist_G(v,t)&\leq \dist_G(v,y)+\dist_G(y,p)+\dist_G(p,t)\\
                &\leq \dist_{G[B(P_j)]}(v,y)+\dist_G(y,p)+\dist_G(p,t)\\
                &<\infty.
  \end{align*}
  We conclude that indeed $\dist_G(s,t)>h$ implies that $\dist_G(s,v)\leq h$ and $\dist_G(v,t)<\infty$.
\end{proof}
With the hitting set $H^*$ computed, we find the transitive closure of $G$ as follows.
Let~$A$~be a $V\times V$ boolean matrix such that $A_{u,v}=1$ iff $\dist_G^h(u,v)\leq h$.
Recall that the values $\dist_G^h(u,v)$ are maintained explicitly
under vertex updates by the data structure of Theorem~\ref{t:h-bounded-submatrix-batch}.
We define an analogous matrix $A^*$ as in Section~\ref{s:tc-reporting}:
\begin{equation*}
  A^*=A\lor \left(A[V,H^*]\cdot (A[H^*,H^*])^n \cdot A[H^*,V]\right),
\end{equation*}
where $\cdot$ denotes the boolean matrix product.
\begin{lemma}
  For any $s,t\in V$, there exists an $s\to t$ path in $G$ iff $A^*_{s,t}=1$.
\end{lemma}
Having proved Lemma~\ref{l:block-hitting}, the proof of the above lemma proceeds completely analogously to that
of Lemma~\ref{l:tc} so we skip it.
To compute the transitive closure of $G$, we simply compute the matrix $A^*$ from the definition
(similarly as in Section~\ref{s:tc-reporting}) via rectangular matrix multiplication.

Let us now analyze the update time. If $\alpha<1$ is such that $|H^*|=\Ot(n^2/(dh))=\Ot(n^\alpha)$,
then the worst-case update time of the data structure can be bounded by:
\begin{equation*}
  \Ot\left(h^2n^2+nd^{\omega-1}+n^{\omega(1,\alpha,1)}\right)\subseteq \Ot\left(h^2n^2+nd^{\omega-1}+\frac{n^{2\omega-2}}{(dh)^{\omega-2}}\right).
\end{equation*}
Note that for $d=n^{1-\eps}$, $h=n^{2\eps}$, the update bound is $\Ot(n^{2+4\eps}+n^{\omega-\eps(\omega-1)}+n^{\omega-\eps(\omega-2)})$.
For any $0<\eps<\frac{\omega-2}{4}$ this bound is polynomially smaller than $n^\omega$.
Such an $\eps$ exists iff $\omega>2$.
\tcdeterministic*
Interestingly, with the currently known upper bounds on the exponents of (rectangular) matrix multiplication,
it is best to set $h=0.15$, $d=0.945$~\cite{Complexity}.  This yields an $O(n^{2.3})$ worst-case update bound
that is slightly worse then the bound obtained in Theorem~\ref{t:tc-reporting}.

\paragraph{Path reporting.} Finally, we remark that given an explicitly recomputed transitive closure
matrix, and by spending additional $O(n^2)$-time preprocessing time per update, reporting paths between arbitrary (reachable) pairs of vertices can be supported in $O(n)$ time.
We now sketch how this is achieved.
\cite[Section~3.2]{Karczmarz0S22} showed that given a topological order of an acyclic digraph $G'$,
for any vertices $s,t$ such that $s$ can reach $t$ in $G'$, some $s\to t$ path
can be constructed by issuing $O(n)$ reachability queries in $G'$.
That is, if the transitive closure matrix of $G'$ is explicitly stored, a path
can be constructed in $O(n)$ time.
Since in our data structure we spend $\Omega(n^2)$ time per update
anyway, after each update in $O(n^2)$ time we can construct a \emph{condensation}
$G'$ of $G$ by contracting the SCCs of $G$ into single vertices.
The condensation $G'$ is a DAG; hence, its topological order can be computed in $O(n^2)$ time. Moreover, the transitive closure
of $G'$ can be easily read from the transitive closure of $G$
by taking a submatrix on arbitrary representatives of the individual SCCs of $G$.
Finally, observe that an $s\to t$ path in $G$ can be found by (1)
finding some path $P'$ from the SCC of $s$ (in $G$) to the SCC of $t$
in the condensation $G'$ in $O(n)$ time, (2) subsequently finding
any $s\to t$ path in the $O(n)$-sized subgraph $P'\cup \gscc$ via graph search.

\section{Randomized improvements}\label{s:randomized}
In this section we show sketch how one can obtain sharper worst-case update bounds for SSSP
and transitive closure if randomization is used.

First of all, in all the components that use path counting, each $h^2$ term in the update/query bounds
can be replaced with a $h$ term in a standard way (see, e.g.,~\cite{Sankowski05}). This is because the used components
only require testing whether the maintained path counts are zero or non-zero.
Hence, instead of using $\widetilde{\Theta}(h)$-bit numbers for path counts,
one can perform counting modulo a random $\Theta(\log{n})$-bit prime number
which is enough to differentiate between zero and non-zero counts
with high probability.
Arithmetic operations on numbers from such a range cost $\Ot(1)$ time as opposed to $\Ot(h)$ time
required in the deterministic path counting data structures.

\subsection{Path-reporting transitive closure}
In the case of path-reporting data structure of Section~\ref{s:tc-reporting}, 
it is not clear how to exploit randomization further (beyond reducing the cost
of path counting) while still maintaining all-pairs single-source reachability trees
against an adaptive adversary. Hence, the obtained randomized worst-case update bound becomes:
\begin{equation*}
  \begin{split}
  \Ot\left(\frac{h n^{\omega(1,\delta,1)+1}}{\Delta b}+\frac{hn^{2+\rho}}{b}\cdot\left(1+\frac{n}{\Delta\tau}\right)+\frac{n^2b}{\Delta}+\Delta h\tau n(hn^{\mu}+\beta)+hn^2\right.\\
    \left.+\frac{hn(n^{\omega(1,\mu,1)-\mu}+n^{1+\mu})}{\beta}+\frac{hn^{\omega(1,\delta,1)+1}}{\beta\Delta}+n^{\omega(1,\alpha,1)}\right).
  \end{split}
\end{equation*}
Using the online balancing tool~\cite{Complexity}, we find that by setting
$\mu=0.368$, $\beta=n^{0.503}$, $b=n^{0.764}$, $d=n^{0.495}$, $h=n^{0.135}$ and $t=n^{0.136}$,
the obtained bound is $O(n^{2.27})$.
That being said, the used improvement has no effect on the theoretical
barrier of the approach: for $\omega<2.12$, $\Ot(n^\omega)$-time trivial
recompute-from-scratch approach is polynomially faster than the above.

\subsection{SSSP}
In the case of single-source shortest paths, we can exploit the standard randomized
hitting set argument~\cite{UY91}, that a randomly
sampled subset $H\subseteq V$ of size $\Theta((n/h)\log{n})$ hits some shortest $s\to t$ path in $G$
for all pairs $s,t$
with $\dist_G(s,t)\geq h$. Crucially, in the dynamic setting, a sampled hitting set remains (w.h.p.) valid
through $\poly(n)$ versions of the evolving graph provided that it is not revealed to the adversary.
In the case of maintaining exact single-source distances in dynamic unweighted digraphs, this assumption is indeed satisfied,
as the answers sought are unique.

The overall framework is the same as in Section~\ref{s:sssp}: we operate in phases of $\Delta$ edge updates,
and maintain global data structures $\dstree$ and $\dsy$, defined exactly as in Section~\ref{s:sssp},
albeit now using the aforementioned randomized path counting.

The main modification we make to the data structure of Section~\ref{s:sssp}
is replacing the preprocessing of Theorem~\ref{t:collection} with the following
variation due to~\cite{AbrahamCK17}.
\begin{lemma}\label{l:collection-rand}{\upshape\cite{AbrahamCK17}}
  Let $G=(V,E)$ be a directed graph. Let $h\in [1,n]$.
 
  Assume one has constructed a data structure storing $G$ and a subset $W\subseteq V$ initially equal to~$V$, and supporting:
  \begin{enumerate}[(1)]
    \item updates removing a vertex $v\in V$ from $W$,
    \item queries computing shortest $h$-bounded paths in $G[W]$ from a query vertex $s$ to \emph{all} $v\in V$.
  \end{enumerate}
  Suppose one can perform updates in $U(n,h)$ \emph{worst-case} time, and queries in $Q(n,h)$ worst-case time.
  These bounds have to hold w.h.p. against an \emph{adaptive} adversary.

  Let $C_0\subseteq V$. Then, in $O(|C_0|\cdot (nh+U(n,h)+Q(n,h)))$ time one can compute a set ${C=\{c_1,\ldots,c_k\}}$,
  where $C_0\subseteq C\subseteq V$ and $k=O(|C_0|)$, and a collection of paths $\Pi$ in $G$, containing at most one path $\pi_{s,t}$
  for all $(s,t)\in C\times V$, such that, with high probability:
  \begin{enumerate}[(a)]
    \item For any $i=1,\ldots,k$ and $t\in V$, $\pi_{c_i,t}\in \Pi$ exists iff $\dist_{G_i}^h(c_i,t)\neq\infty$ and then $\pi_{c_i,t}$ is a shortest $h$-bounded $c_i\to t$ path in the graph $G_i:=G-\{c_1,\ldots,c_{i-1}\}$.
    \item For any $u\in V$, there exist at most $O(hn\log{n})$ paths $\pi_{s,t}\in \Pi$ with $u\in V(\pi_{s,t})$.
  \end{enumerate}
\end{lemma}
\begin{proof}[Proof sketch.]
  The preprocessing differs from that in Lemma~\ref{l:collection-thres} in the way how the subsequent
  sources are picked. Whereas in Lemma~\ref{l:collection-thres} sources (from all of $V$) are picked in arbitrary
  order, in the construction of~\cite[Section~3.1]{AbrahamCK17}, one alternates between picking
  an arbitrary unprocessed element of $C_0$, and the most-congested vertex (that is, lying on the maximum
  number of paths picked so far) from $V\setminus C$.
  This way, only $O(|C_0|)$ vertices are added to $W$ and $O(|C_0|)$ queries are issued.
  The $O(hn\log{n})$ bound in item~(b) follows via an analysis of a greedy load balancing
  technique~\cite{LevcopoulosO88}. This argument has been originally used in this context by~\cite{Thorup05}.
  See~\cite{AbrahamCK17} for details.
\end{proof}
More specifically, the congested set/paths collection pair $(C,\Pi)$ of Theorem~\ref{t:collection}
is replaced with the corresponding pair $(C,\Pi)$ produced by Lemma~\ref{l:collection-rand}
with $C_0:=H$, where $H$ is the aforementioned randomly sampled hitting set.
Importantly, the obtained set $C$ has size $\Ot(n/h)$ this way, and the
collection $\Pi$ has $\Ot(n^2/h)$ paths (as opposed to $\Theta(n^2)$ in Theorem~\ref{t:collection}).

For computing the pair $(C,\Pi)$, we use the data structure $\dstree$. Consequently, the worst-case time cost
(per update) incurred by the construction of $(C,\Pi)$ is now:
\begin{equation*}
  \Ot\left(\frac{1}{\Delta}\cdot\frac{n}{h}\cdot \left(\frac{hn^{1+\omega(1,\mu,1)-\mu}}{b}+\frac{hn^{2+\mu}}{b}+nb\right)\right).
\end{equation*}
For $\mu=\rho$, this becomes:
\begin{equation*}
  \Ot\left(\frac{n^{3+\rho}}{b\Delta}+\frac{n^2b}{h\Delta}\right).
\end{equation*}
Recall that besides that, the cost of maintaining $\dstree$ under single-edge updates
to $G$ is $\Ot(hn^{2+\rho}/b)$ per update.

We also redefine the contents of the maintained set $Y$ throughout the phase.
Recall that the values $\dist^h_G(s,t)$ are maintained for all $(s,t)\in Y$ by the data
structure $\dsy$ with $\nu$ set to $\rho\approx 0.529$.
At all times, the set $Y$ contains:
\begin{itemize}
  \item all pairs $(s,t)\in C\times V$ such that $\pi_{s,t}\in \Pi_D$,
  \item all pairs $(u,v)\in (D\times V)\cup (V\times D)$.
\end{itemize}

Observe that by Lemma~\ref{l:collection-rand}, the set $Y$ grows by $\Ot(hn)$ new pairs per update in the worst case, 
and thus can be appropriately extended in $\Ot(hn^{1+\rho})$ time per update.
Moreover, $|Y|=\Ot(\Delta hn)$ at all times and thus $\dsy$ requires $\Ot(\Delta h^2n)$ additional time per update
to recompute the values $\dist_G^h(s,t)$ for all $(s,t)\in Y$.

The summarized worst-case update time of all the components of the data structure is:
\begin{equation*}
  \Ot\left(\frac{hn^{2+\rho}}{b}+\Delta h^2n+\frac{n^{3+\rho}}{b\Delta}+\frac{n^2b}{h\Delta}+h^2n^{1+\rho}\right).
\end{equation*}

For answering single-source distance queries, we use a slightly simpler auxiliary graph $X$, defined similarly
as in Section~\ref{s:sssp}.
Suppose a query for distances from a source vertex $s\in V$ is issued.
Let $Y_s$ be the set $Y$ augmented with temporary pairs $(s,t)$ for all $t\in V$.
The graph $X$ contains:
\begin{itemize}
  \item an edge $uv$ of weight $|\pi_{u,v}|$ for all $\pi_{u,v} \in \Pi\setminus \Pi_D$,
  \item an edge $uv$ of weight $\dist_G^h(u,v)$ for all $(u,v)\in Y_s$.
\end{itemize}
Again, the graph~$X$ has $O(n^2/h+|Y|)=\Ot(n^2/h+\Delta h n)$ edges. The query procedure runs
Dijkstra's algorithm from $s$ on the graph $X$.
This is correct by the following lemma.

\begin{lemma}
  With high probability, for any $t\in V$, $\dist_G(s,t)=\dist_X(s,t)$.
\end{lemma}
\begin{proof}
  Note that for all $uv\in E(X)$, $\wei_X(uv)\geq \dist_G(uv)$, which establishes $\dist_G(s,t)\leq \dist_X(s,t)$ easily.

  Let us now prove that $\dist_X(s,t)\leq \dist_G(s,t)$.
  First of all, let $G_0$ be the graph at the beginning of the phase.
  Note that by the definition of $D$, we have $G_0-D=G-D\subseteq G$.

  Consider the shortest $s\to t$ path $P$ in $G$; if there
  are many, pick the unique one that is lexicographically smallest.
  Split $P$ into maximal segments $P_1\cdot \ldots\cdot P_k$ not passing through any
  of the vertices $C\cup D$. In other words, for each $P_i=u_i\to v_i$, we have $V(P_i)\cap (C\cup D)\subseteq \{u_i,v_i\}$
  and $u_i\in C\cup D\cup \{s\}$.

  Consider some $P_i$. $P_i$ is a lexicographically shortest path in $G$.
  We have $|P_i|\leq h$ with high probability, since $H\subseteq C$
  hits each of the $O(n^2)$ lexicographically shortest paths that have more than $h$ hops in every
  of $O(\poly{n})$ versions of $G$ independent of $H$.
  As a result, we have
  \begin{equation*}
    \dist_G(u_i,v_i)=|P_i|=\dist_G^h(u_i,v_i).
  \end{equation*}
  We now prove that $\wei_{X}(u_iv_i)=|P_i|$. This will imply that there
  exists an $s\to t$ path in $X$ of weight $\sum_{i}\wei_{X}(u_iv_i)\leq \sum_{i}|P_i|=\dist_G(s,t)$.
  If $(u_i,v_i)\in Y_s$, then $\wei_X(u_iv_i)=\dist_G^h(u_i,v_i)$ by construction.

  Suppose that $(u_i,v_i)\notin Y_s$. Then we simultaneously have $u_i\in C\setminus D$, $v_i\in V\setminus D$, and
  either $\pi_{u_i,v_i}$ does not exist in $\Pi$ or $\pi_{u_i,v_i}\in \Pi\setminus \Pi_D$, i.e., $\pi_{u_i,v_i}$ does not pass through a vertex of $D$.
  Let $u_i=c_j\in C$, where $c_j$ and $G_j\subseteq G_0$ are defined as in Theorem~\ref{t:collection}.
  Since $P_i$ is internally disjoint with $C\cup D$,
  and $v_i\notin D$, $P_i$ contains no edges inserted in the current phase.
  Hence, $P_i\subseteq G_j$. 
  As a result, ${\dist_{G_j}^h(u_i,v_i)\leq |P_i|<\infty}$.
  This proves that $\pi_{u_i,v_i}=\pi_{c_j,v_i}$ exists in $\Pi$.
  But we also have $\pi_{u_i,v_i}\subseteq G_j-D\subseteq G_0-D\subseteq G-D$, so $\dist_{G}(u_i,v_i)\leq |\pi_{u_i,v_i}|$.
  Finally, by Theorem~\ref{t:collection}, we obtain:
  \begin{equation*}
    \dist_G(u_i,v_i)\leq |\pi_{u_i,v_i}|=\dist_{G_j}^h(u_i,v_i)\leq |P_i|=\dist_G(u_i,v_i).
  \end{equation*}
  This proves that $\wei_G(u_i,v_i)=\dist_G(u_i,v_i)=|P_i|$ also when $(u_i,v_i)\notin Y_s$.
\end{proof}

We obtain the following combined worst-case update/query bound:
\begin{equation*}
  \Ot\left(\frac{hn^{2+\rho}}{b}+\Delta h^2n+\frac{n^{3+\rho}}{b\Delta}+\frac{n^2b}{h\Delta}+h^2n^{1+\rho}+\frac{n^2}{h}\right).
\end{equation*}

Note that there is no point in setting $b<\Delta$,
as increasing $b$ to $\Delta$ would not make harm because of the $n^2/h$ term.
So $b\geq \Delta$.
If we aim at equal update and query time, there is no point in putting $b>\Delta$ either.
For $b=\Delta$ we need to balance $\Delta h^2n$, $\frac{n^{3+\rho}}{\Delta^2}$ and $n^2/h$.
We get $\Delta=n/h^3$ and finally we balance $n^{1+\rho}h^6$ and $n^2/h$.
For $h=n^{(1-\rho)/7}$,
the update/query time is $O(n^{(13+\rho)/7})=O(n^{1.933})$.
We have thus proved.

\begin{theorem}\label{t:ss-distances-randomized}
  Let $G=(V,E)$ be unweighted digraph. There exists a Monte Carlo randomized data structure supporting
  fully dynamic single-edge updates to $G$ and computing exact distances from an arbitrary query source $s\in V$
  to all other vertices in $G$
  in $O(n^{1.933})$ worst-case time.

  If $\omega=2$, the bound becomes $\Ot(n^{2-1/14})$.
\end{theorem}

\bibliographystyle{alpha}
\bibliography{references}

\newcommand{\etalchar}[1]{$^{#1}$}
\begin{thebibliography}{AGMN92}

\bibitem[ABF23]{AbboudBF22}
Amir Abboud, Karl Bringmann, and Nick Fischer.
\newblock Stronger 3-sum lower bounds for approximate distance oracles via
  additive combinatorics.
\newblock In {\em Proceedings of the 55th Annual {ACM} Symposium on Theory of
  Computing, {STOC} 2023, Orlando, FL, USA, June 20-23, 2023}, pages 391--404.
  {ACM}, 2023.

\bibitem[ABKZ22]{AbboudBKZ22}
Amir Abboud, Karl Bringmann, Seri Khoury, and Or~Zamir.
\newblock Hardness of approximation in p via short cycle removal: cycle
  detection, distance oracles, and beyond.
\newblock In {\em {STOC}}, pages 1487--1500. {ACM}, 2022.

\bibitem[ACG12]{AbrahamCG12}
Ittai Abraham, Shiri Chechik, and Cyril Gavoille.
\newblock Fully dynamic approximate distance oracles for planar graphs via
  forbidden-set distance labels.
\newblock In {\em {STOC}}, pages 1199--1218. {ACM}, 2012.

\bibitem[ACK17]{AbrahamCK17}
Ittai Abraham, Shiri Chechik, and Sebastian Krinninger.
\newblock Fully dynamic all-pairs shortest paths with worst-case update-time
  revisited.
\newblock In {\em Proceedings of the Twenty-Eighth Annual {ACM-SIAM} Symposium
  on Discrete Algorithms, {SODA} 2017, Barcelona, Spain, Hotel Porta Fira,
  January 16-19}, pages 440--452. {SIAM}, 2017.

\bibitem[ACT14]{AbrahamCT14}
Ittai Abraham, Shiri Chechik, and Kunal Talwar.
\newblock Fully dynamic all-pairs shortest paths: Breaking the o(n) barrier.
\newblock In {\em {APPROX-RANDOM}}, volume~28 of {\em LIPIcs}, pages 1--16.
  Schloss Dagstuhl - Leibniz-Zentrum fuer Informatik, 2014.

\bibitem[AGMN92]{AlonGMN92}
Noga Alon, Zvi Galil, Oded Margalit, and Moni Naor.
\newblock Witnesses for boolean matrix multiplication and for shortest paths.
\newblock In {\em 33rd Annual Symposium on Foundations of Computer Science,
  Pittsburgh, Pennsylvania, USA, 24-27 October 1992}, pages 417--426. {IEEE}
  Computer Society, 1992.

\bibitem[AHR{\etalchar{+}}19]{AnconaHRWW19}
Bertie Ancona, Monika Henzinger, Liam Roditty, Virginia~Vassilevska Williams,
  and Nicole Wein.
\newblock Algorithms and hardness for diameter in dynamic graphs.
\newblock In {\em {ICALP}}, volume 132 of {\em LIPIcs}, pages 13:1--13:14.
  Schloss Dagstuhl - Leibniz-Zentrum f{\"{u}}r Informatik, 2019.

\bibitem[AIMN91]{AusielloIMN91}
Giorgio Ausiello, Giuseppe~F. Italiano, Alberto Marchetti{-}Spaccamela, and
  Umberto Nanni.
\newblock Incremental algorithms for minimal length paths.
\newblock {\em J. Algorithms}, 12(4):615--638, 1991.

\bibitem[AvdB23]{AlokhinaB23}
Anastasiia Alokhina and Jan van~den Brand.
\newblock Fully dynamic shortest path reporting against an adaptive adversary.
\newblock {\em CoRR}, abs/2304.07403, 2023.

\bibitem[AW14]{AbboudW14}
Amir Abboud and Virginia~Vassilevska Williams.
\newblock Popular conjectures imply strong lower bounds for dynamic problems.
\newblock In {\em {FOCS}}, pages 434--443. {IEEE} Computer Society, 2014.

\bibitem[AW21]{AlmanW21}
Josh Alman and Virginia~Vassilevska Williams.
\newblock A refined laser method and faster matrix multiplication.
\newblock In {\em Proceedings of the 2021 {ACM-SIAM} Symposium on Discrete
  Algorithms, {SODA} 2021, Virtual Conference, January 10 - 13, 2021}, pages
  522--539. {SIAM}, 2021.

\bibitem[BC16]{BernsteinC16}
Aaron Bernstein and Shiri Chechik.
\newblock Deterministic decremental single source shortest paths: beyond the
  o(mn) bound.
\newblock In {\em {STOC}}, pages 389--397. {ACM}, 2016.

\bibitem[BC17]{BernsteinC17}
Aaron Bernstein and Shiri Chechik.
\newblock Deterministic partially dynamic single source shortest paths for
  sparse graphs.
\newblock In {\em {SODA}}, pages 453--469. {SIAM}, 2017.

\bibitem[Ber09]{Bernstein09}
Aaron Bernstein.
\newblock Fully dynamic {(2} + epsilon) approximate all-pairs shortest paths
  with fast query and close to linear update time.
\newblock In {\em {FOCS}}, pages 693--702. {IEEE} Computer Society, 2009.

\bibitem[Ber16]{Bernstein16}
Aaron Bernstein.
\newblock Maintaining shortest paths under deletions in weighted directed
  graphs.
\newblock {\em {SIAM} J. Comput.}, 45(2):548--574, 2016.
\newblock Announced at STOC'13.

\bibitem[Ber17]{Bernstein17}
Aaron Bernstein.
\newblock Deterministic partially dynamic single source shortest paths in
  weighted graphs.
\newblock In {\em {ICALP}}, volume~80 of {\em LIPIcs}, pages 44:1--44:14.
  Schloss Dagstuhl - Leibniz-Zentrum f{\"{u}}r Informatik, 2017.

\bibitem[BFN22]{BrandFN22}
Jan van~den Brand, Sebastian Forster, and Yasamin Nazari.
\newblock Fast deterministic fully dynamic distance approximation.
\newblock In {\em 63rd {IEEE} Annual Symposium on Foundations of Computer
  Science, {FOCS} 2022, Denver, CO, USA, October 31 - November 3, 2022}, pages
  1011--1022. {IEEE}, 2022.

\bibitem[BGS20]{BernsteinGS20}
Aaron Bernstein, Maximilian~Probst Gutenberg, and Thatchaphol Saranurak.
\newblock Deterministic decremental reachability, scc, and shortest paths via
  directed expanders and congestion balancing.
\newblock In {\em {FOCS}}, pages 1123--1134. {IEEE}, 2020.

\bibitem[BGS21]{BernsteinGS21}
Aaron Bernstein, Maximilian~Probst Gutenberg, and Thatchaphol Saranurak.
\newblock Deterministic decremental {SSSP} and approximate min-cost flow in
  almost-linear time.
\newblock In {\em {FOCS}}, pages 1000--1008. {IEEE}, 2021.

\bibitem[BGW20]{BernsteinGW20}
Aaron Bernstein, Maximilian~Probst Gutenberg, and Christian Wulff{-}Nilsen.
\newblock Near-optimal decremental {SSSP} in dense weighted digraphs.
\newblock In {\em {FOCS}}, pages 1112--1122. {IEEE}, 2020.

\bibitem[BHG{\etalchar{+}}21]{BergamaschiHGWW21}
Thiago Bergamaschi, Monika Henzinger, Maximilian~Probst Gutenberg,
  Virginia~Vassilevska Williams, and Nicole Wein.
\newblock New techniques and fine-grained hardness for dynamic near-additive
  spanners.
\newblock In {\em {SODA}}, pages 1836--1855. {SIAM}, 2021.

\bibitem[BHS07]{BaswanaHS07}
Surender Baswana, Ramesh Hariharan, and Sandeep Sen.
\newblock Improved decremental algorithms for maintaining transitive closure
  and all-pairs shortest paths.
\newblock {\em J. Algorithms}, 62(2):74--92, 2007.

\bibitem[BN19]{BrandN19}
Jan van~den Brand and Danupon Nanongkai.
\newblock Dynamic approximate shortest paths and beyond: Subquadratic and
  worst-case update time.
\newblock In {\em 60th {IEEE} Annual Symposium on Foundations of Computer
  Science, {FOCS} 2019, Baltimore, Maryland, USA, November 9-12, 2019}, pages
  436--455. {IEEE} Computer Society, 2019.

\bibitem[BNS19]{BrandNS19}
Jan van~den Brand, Danupon Nanongkai, and Thatchaphol Saranurak.
\newblock Dynamic matrix inverse: Improved algorithms and matching conditional
  lower bounds.
\newblock In {\em {FOCS}}, pages 456--480. {IEEE} Computer Society, 2019.

\bibitem[BR11]{BernsteinR11}
Aaron Bernstein and Liam Roditty.
\newblock Improved dynamic algorithms for maintaining approximate shortest
  paths under deletions.
\newblock In {\em {SODA}}, pages 1355--1365. {SIAM}, 2011.

\bibitem[Bra]{Complexity}
Jan van~den Brand.
\newblock Complexity term balancer.
\newblock \url{www.ocf.berkeley.edu/~vdbrand/complexity/}.
\newblock Tool to balance complexity terms depending on fast matrix
  multiplication.

\bibitem[BS19]{BrandS19}
Jan van~den Brand and Thatchaphol Saranurak.
\newblock Sensitive distance and reachability oracles for large batch updates.
\newblock In {\em {FOCS}}, pages 424--435. {IEEE} Computer Society, 2019.

\bibitem[Chu21]{Chuzhoy21}
Julia Chuzhoy.
\newblock Decremental all-pairs shortest paths in deterministic near-linear
  time.
\newblock In {\em {STOC}}, pages 626--639. {ACM}, 2021.

\bibitem[CK19]{ChuzhoyK19}
Julia Chuzhoy and Sanjeev Khanna.
\newblock A new algorithm for decremental single-source shortest paths with
  applications to vertex-capacitated flow and cut problems.
\newblock In {\em {STOC}}, pages 389--400. {ACM}, 2019.

\bibitem[CS21]{ChuzhoyS21}
Julia Chuzhoy and Thatchaphol Saranurak.
\newblock Deterministic algorithms for decremental shortest paths via layered
  core decomposition.
\newblock In {\em {SODA}}, pages 2478--2496. {SIAM}, 2021.

\bibitem[CZ21]{ChechikZ21}
Shiri Chechik and Tianyi Zhang.
\newblock Incremental single source shortest paths in sparse digraphs.
\newblock In {\em {SODA}}, pages 2463--2477. {SIAM}, 2021.

\bibitem[CZ23]{ChechikZ23}
Shiri Chechik and Tianyi Zhang.
\newblock Faster deterministic worst-case fully dynamic all-pairs shortest
  paths via decremental hop-restricted shortest paths.
\newblock In {\em Proceedings of the 2023 {ACM-SIAM} Symposium on Discrete
  Algorithms, {SODA} 2023, Florence, Italy, January 22-25, 2023}, pages 87--99.
  {SIAM}, 2023.

\bibitem[DFNV22]{DoryFNV22}
Michal Dory, Sebastian Forster, Yasamin Nazari, and Tijn~de Vos.
\newblock New tradeoffs for decremental approximate all-pairs shortest paths.
\newblock {\em CoRR}, abs/2211.01152, 2022.

\bibitem[DI01]{DemetrescuI01}
Camil Demetrescu and Giuseppe~F. Italiano.
\newblock Fully dynamic all pairs shortest paths with real edge weights.
\newblock In {\em {FOCS}}, pages 260--267. {IEEE} Computer Society, 2001.

\bibitem[DI02]{DemetrescuI02}
Camil Demetrescu and Giuseppe~F. Italiano.
\newblock Improved bounds and new trade-offs for dynamic all pairs shortest
  paths.
\newblock In {\em {ICALP}}, volume 2380 of {\em Lecture Notes in Computer
  Science}, pages 633--643. Springer, 2002.

\bibitem[DI04]{DemetrescuI04}
Camil Demetrescu and Giuseppe~F. Italiano.
\newblock A new approach to dynamic all pairs shortest paths.
\newblock {\em J. {ACM}}, 51(6):968--992, 2004.
\newblock Announced at STOC'03.

\bibitem[DI05]{DemetrescuI00}
Camil Demetrescu and Giuseppe~F. Italiano.
\newblock Trade-offs for fully dynamic transitive closure on dags: breaking
  through the o(n\({}^{\mbox{2}}\) barrier.
\newblock {\em J. {ACM}}, 52(2):147--156, 2005.

\bibitem[DI08]{DemetrescuI08}
Camil Demetrescu and Giuseppe~F. Italiano.
\newblock Mantaining dynamic matrices for fully dynamic transitive closure.
\newblock {\em Algorithmica}, 51(4):387--427, 2008.

\bibitem[DWZ22]{fmm}
Ran Duan, Hongxun Wu, and Renfei Zhou.
\newblock Faster matrix multiplication via asymmetric hashing.
\newblock {\em CoRR}, abs/2210.10173, 2022.

\bibitem[EFGW21]{EvaldFGW21}
Jacob Evald, Viktor Fredslund{-}Hansen, Maximilian~Probst Gutenberg, and
  Christian Wulff{-}Nilsen.
\newblock Decremental {APSP} in unweighted digraphs versus an adaptive
  adversary.
\newblock In {\em {ICALP}}, volume 198 of {\em LIPIcs}, pages 64:1--64:20.
  Schloss Dagstuhl - Leibniz-Zentrum f{\"{u}}r Informatik, 2021.

\bibitem[ES81]{EvenS81}
Shimon Even and Yossi Shiloach.
\newblock An on-line edge-deletion problem.
\newblock {\em J. {ACM}}, 28(1):1--4, 1981.

\bibitem[FGH21]{ForsterGH21}
Sebastian Forster, Gramoz Goranci, and Monika Henzinger.
\newblock Dynamic maintenance of low-stretch probabilistic tree embeddings with
  applications.
\newblock In {\em {SODA}}, pages 1226--1245. {SIAM}, 2021.

\bibitem[FGNS23]{ForsterGNS23}
Sebastian Forster, Gramoz Goranci, Yasamin Nazari, and Antonis Skarlatos.
\newblock Bootstrapping dynamic distance oracles.
\newblock In {\em 31st Annual European Symposium on Algorithms, {ESA} 2023,
  September 4-6, 2023, Amsterdam, The Netherlands}, volume 274 of {\em LIPIcs},
  pages 50:1--50:16. Schloss Dagstuhl - Leibniz-Zentrum f{\"{u}}r Informatik,
  2023.

\bibitem[FMNZ01]{FrigioniMNZ01}
Daniele Frigioni, Tobias Miller, Umberto Nanni, and Christos~D. Zaroliagis.
\newblock An experimental study of dynamic algorithms for transitive closure.
\newblock {\em {ACM} J. Exp. Algorithmics}, 6:9, 2001.

\bibitem[GR21]{GuR21}
Yong Gu and Hanlin Ren.
\newblock Constructing a distance sensitivity oracle in o(n{\^{}}2.5794 {M)}
  time.
\newblock In {\em {ICALP}}, volume 198 of {\em LIPIcs}, pages 76:1--76:20.
  Schloss Dagstuhl - Leibniz-Zentrum f{\"{u}}r Informatik, 2021.

\bibitem[GU18]{GallU18}
Francois~Le Gall and Florent Urrutia.
\newblock Improved rectangular matrix multiplication using powers of the
  coppersmith-winograd tensor.
\newblock In {\em Proceedings of the Twenty-Ninth Annual {ACM-SIAM} Symposium
  on Discrete Algorithms, {SODA} 2018, New Orleans, LA, USA, January 7-10,
  2018}, pages 1029--1046. {SIAM}, 2018.

\bibitem[GW20a]{GutenbergW20a}
Maximilian~Probst Gutenberg and Christian Wulff{-}Nilsen.
\newblock Decremental {SSSP} in weighted digraphs: Faster and against an
  adaptive adversary.
\newblock In {\em {SODA}}, pages 2542--2561. {SIAM}, 2020.

\bibitem[GW20b]{GutenbergW20b}
Maximilian~Probst Gutenberg and Christian Wulff{-}Nilsen.
\newblock Fully-dynamic all-pairs shortest paths: Improved worst-case time and
  space bounds.
\newblock In {\em {SODA}}, pages 2562--2574. {SIAM}, 2020.

\bibitem[GWW20]{GutenbergWW20}
Maximilian~Probst Gutenberg, Virginia~Vassilevska Williams, and Nicole Wein.
\newblock New algorithms and hardness for incremental single-source shortest
  paths in directed graphs.
\newblock In {\em {STOC}}, pages 153--166. {ACM}, 2020.

\bibitem[HK95]{HenzingerK95}
Monika~Rauch Henzinger and Valerie King.
\newblock Fully dynamic biconnectivity and transitive closure.
\newblock In {\em {FOCS}}, pages 664--672. {IEEE} Computer Society, 1995.

\bibitem[HKN13]{HenzingerKN-ICALP13}
Monika Henzinger, Sebastian Krinninger, and Danupon Nanongkai.
\newblock Sublinear-time maintenance of breadth-first spanning tree in
  partially dynamic networks.
\newblock In {\em {ICALP} {(2)}}, volume 7966 of {\em Lecture Notes in Computer
  Science}, pages 607--619. Springer, 2013.

\bibitem[HKN14]{HenzingerKN14}
Monika Henzinger, Sebastian Krinninger, and Danupon Nanongkai.
\newblock Sublinear-time decremental algorithms for single-source reachability
  and shortest paths on directed graphs.
\newblock In {\em {STOC}}, pages 674--683. {ACM}, 2014.

\bibitem[HKN16]{HenzingerKN13}
Monika Henzinger, Sebastian Krinninger, and Danupon Nanongkai.
\newblock Dynamic approximate all-pairs shortest paths: Breaking the {O(mn)}
  barrier and derandomization.
\newblock {\em {SIAM} J. Comput.}, 45(3):947--1006, 2016.
\newblock Announced at FOCS'13.

\bibitem[HKNS15]{HenzingerKNS15}
Monika Henzinger, Sebastian Krinninger, Danupon Nanongkai, and Thatchaphol
  Saranurak.
\newblock Unifying and strengthening hardness for dynamic problems via the
  online matrix-vector multiplication conjecture.
\newblock In {\em {STOC}}, pages 21--30. {ACM}, 2015.

\bibitem[HP98]{HuangP98}
Xiaohan Huang and Victor~Y. Pan.
\newblock Fast rectangular matrix multiplication and applications.
\newblock {\em J. Complex.}, 14(2):257--299, 1998.

\bibitem[Ita86]{Italiano86}
Giuseppe~F. Italiano.
\newblock Amortized efficiency of a path retrieval data structure.
\newblock {\em Theor. Comput. Sci.}, 48(3):273--281, 1986.

\bibitem[JX22]{JinX22}
Ce~Jin and Yinzhan Xu.
\newblock Tight dynamic problem lower bounds from generalized {BMM} and omv.
\newblock In {\em {STOC} '22: 54th Annual {ACM} {SIGACT} Symposium on Theory of
  Computing, Rome, Italy, June 20 - 24, 2022}, pages 1515--1528. {ACM}, 2022.

\bibitem[Kin99]{King99}
Valerie King.
\newblock Fully dynamic algorithms for maintaining all-pairs shortest paths and
  transitive closure in digraphs.
\newblock In {\em {FOCS}}, pages 81--91. {IEEE} Computer Society, 1999.

\bibitem[KL19]{KarczmarzL19}
Adam Karczmarz and Jakub Lacki.
\newblock Reliable hubs for partially-dynamic all-pairs shortest paths in
  directed graphs.
\newblock In {\em {ESA}}, volume 144 of {\em LIPIcs}, pages 65:1--65:15.
  Schloss Dagstuhl - Leibniz-Zentrum f{\"{u}}r Informatik, 2019.

\bibitem[KL20]{KarczmarzL20}
Adam Karczmarz and Jakub Lacki.
\newblock Simple label-correcting algorithms for partially dynamic approximate
  shortest paths in directed graphs.
\newblock In {\em 3rd Symposium on Simplicity in Algorithms, {SOSA} 2020, Salt
  Lake City, UT, USA, January 6-7, 2020}, pages 106--120. {SIAM}, 2020.

\bibitem[KMG22]{KyngMG22}
Rasmus Kyng, Simon Meierhans, and Maximilian~Probst Gutenberg.
\newblock Incremental {SSSP} for sparse digraphs beyond the hopset barrier.
\newblock In {\em {SODA}}, pages 3452--3481. {SIAM}, 2022.

\bibitem[KMS22]{Karczmarz0S22}
Adam Karczmarz, Anish Mukherjee, and Piotr Sankowski.
\newblock Subquadratic dynamic path reporting in directed graphs against an
  adaptive adversary.
\newblock In {\em {STOC} '22: 54th Annual {ACM} {SIGACT} Symposium on Theory of
  Computing, Rome, Italy, June 20 - 24, 2022}, pages 1643--1656. {ACM}, 2022.

\bibitem[KS02]{KingS02}
Valerie King and Garry Sagert.
\newblock A fully dynamic algorithm for maintaining the transitive closure.
\newblock {\em J. Comput. Syst. Sci.}, 65(1):150--167, 2002.
\newblock Announced at STOC'99.

\bibitem[KS23]{KarczmarzS23}
Adam Karczmarz and Piotr Sankowski.
\newblock Sensitivity and dynamic distance oracles via generic matrices and
  {F}robenius form.
\newblock {\em CoRR}, abs/2308.08870, 2023.

\bibitem[Lac11]{Lacki11}
Jakub Lacki.
\newblock Improved deterministic algorithms for decremental transitive closure
  and strongly connected components.
\newblock In {\em {SODA}}, pages 1438--1445. {SIAM}, 2011.

\bibitem[LO88]{LevcopoulosO88}
Christos Levcopoulos and Mark~H. Overmars.
\newblock A balanced search tree with \emph{ {O} } {(1)} worst-case update
  time.
\newblock {\em Acta Informatica}, 26(3):269--277, 1988.

\bibitem[Mao23]{Mao23}
Xiao Mao.
\newblock Fully-dynamic all-pairs shortest paths: Likely optimal worst-case
  update time.
\newblock {\em CoRR}, abs/2306.02662, 2023.

\bibitem[Mun71]{Munro71}
J.~Ian Munro.
\newblock Efficient determination of the transitive closure of a directed
  graph.
\newblock {\em Inf. Process. Lett.}, 1(2):56--58, 1971.

\bibitem[Rod08]{Roditty08}
Liam Roditty.
\newblock A faster and simpler fully dynamic transitive closure.
\newblock {\em {ACM} Trans. Algorithms}, 4(1):6:1--6:16, 2008.

\bibitem[RZ11]{RodittyZ11}
Liam Roditty and Uri Zwick.
\newblock On dynamic shortest paths problems.
\newblock {\em Algorithmica}, 61(2):389--401, 2011.

\bibitem[RZ12]{RodittyZ12}
Liam Roditty and Uri Zwick.
\newblock Dynamic approximate all-pairs shortest paths in undirected graphs.
\newblock {\em {SIAM} J. Comput.}, 41(3):670--683, 2012.
\newblock Announced at FOCS'04.

\bibitem[San04]{Sankowski04}
Piotr Sankowski.
\newblock Dynamic transitive closure via dynamic matrix inverse (extended
  abstract).
\newblock In {\em 45th Symposium on Foundations of Computer Science {(FOCS}
  2004), 17-19 October 2004, Rome, Italy, Proceedings}, pages 509--517. {IEEE}
  Computer Society, 2004.

\bibitem[San05]{Sankowski05}
Piotr Sankowski.
\newblock Subquadratic algorithm for dynamic shortest distances.
\newblock In {\em Computing and Combinatorics, 11th Annual International
  Conference, {COCOON} 2005, Kunming, China, August 16-29, 2005, Proceedings},
  volume 3595 of {\em Lecture Notes in Computer Science}, pages 461--470.
  Springer, 2005.

\bibitem[San07]{Sankowski07}
Piotr Sankowski.
\newblock Faster dynamic matchings and vertex connectivity.
\newblock In {\em {SODA}}, pages 118--126. {SIAM}, 2007.

\bibitem[SM10]{SankowskiM10}
Piotr Sankowski and Marcin Mucha.
\newblock Fast dynamic transitive closure with lookahead.
\newblock {\em Algorithmica}, 56(2):180--197, 2010.

\bibitem[Tar72]{Tarjan72}
Robert~Endre Tarjan.
\newblock Depth-first search and linear graph algorithms.
\newblock {\em {SIAM} J. Comput.}, 1(2):146--160, 1972.

\bibitem[Tho04]{Thorup04}
Mikkel Thorup.
\newblock Fully-dynamic all-pairs shortest paths: Faster and allowing negative
  cycles.
\newblock In {\em {SWAT}}, volume 3111 of {\em Lecture Notes in Computer
  Science}, pages 384--396. Springer, 2004.

\bibitem[Tho05]{Thorup05}
Mikkel Thorup.
\newblock Worst-case update times for fully-dynamic all-pairs shortest paths.
\newblock In {\em Proceedings of the 37th Annual {ACM} Symposium on Theory of
  Computing, Baltimore, MD, USA, May 22-24, 2005}, pages 112--119. {ACM}, 2005.

\bibitem[UY91]{UY91}
Jeffrey~D. Ullman and Mihalis Yannakakis.
\newblock High-probability parallel transitive-closure algorithms.
\newblock {\em {SIAM} J. Comput.}, 20(1):100--125, 1991.

\bibitem[WXXZ23]{WilliamsXXZ23}
Virginia~Vassilevska Williams, Yinzhan Xu, Zixuan Xu, and Renfei Zhou.
\newblock New bounds for matrix multiplication: from alpha to omega.
\newblock {\em CoRR}, abs/2307.07970, 2023.

\bibitem[Zwi02]{Zwick02}
Uri Zwick.
\newblock All pairs shortest paths using bridging sets and rectangular matrix
  multiplication.
\newblock {\em J. {ACM}}, 49(3):289--317, 2002.

\end{thebibliography}

\appendix

\section{Algebraic data structure (\Cref{t:h-bounded-submatrix-batch})}\label{s:algebraic}
In this section we show the algebraic subroutine used by our dynamic graph algorithms.
The dynamic algorithm is based on the reduction from dynamic distance to dynamic matrix inverse from \cite{Sankowski05} and it's deterministic equivalent from \cite{BrandFN22}.
The dynamic matrix inverse algorithm will imply the following dynamic distance result.

\algebraic*

\Cref{t:h-bounded-submatrix-batch} is obtained from the following dynamic matrix inverse data structure.
Note that unlike \Cref{t:h-bounded-submatrix-batch}, \Cref{lem:matrix_inverse} measures the complexity in the number of arithmetic operations. This is because depending on what field or ring the matrix $\mM$ is from, one iteration may take more than $O(1)$ time.
In particular when we later reduce dynamic distance to dynamic matrix inverse in order to prove \Cref{t:h-bounded-submatrix-batch}, one arithmetic operation will take $\tilde{O}(h^2)$ time.

\begin{lemma}\label{lem:matrix_inverse}
	Let $\mM$ be a dynamic $n\times n$ matrix receiving entry updates that is promised to stay non-singular throughout all updates.
	Let $Y\subseteq [n]\times [n]$ be a dynamic set.
	For any $\mu \in [0,1]$, there exists a deterministic data structure maintaining the values $(\mM^{-1})_{s,t}$ for all $(s,t)\in Y$
	and supporting:
	\begin{itemize}
		\item single-entry updates to $\mM$ in $O(n^{\omega(1,\mu,1)-\mu}+n^{1+\mu}+|Y|)$ operations in worst-case,
		\item rank $n^\delta$ updates to $\mM$ for any $\delta \in [0,1]$ in $O(n^{\omega(1,\delta,1)})$ operations in worst-case, 
		\item single-element insertions to $Y$ in $O(n^\mu)$ operations,
		\item removing elements from $Y$ in $O(1)$ time.
	\end{itemize}
\end{lemma}

\begin{proof}
	The data structure internally maintains $n\times n$ matrices $\mN$ and $n \times k$ matrices $\mL,\mR$ (with $k \le n^\mu$) such that
	\begin{align}\label{eq:inverse}
	\mM^{-1} = \mN - \mL \mR^\top.
	\end{align}
	Further, the data structure explicitly maintains the values $\mM^{-1}_{s,t}$ for all $s,t \in Y$.
	\paragraph{Initialization.}
	During initialization, the data structure computes $\mM^{-1}$ in $O(n^\omega)$ operations and stores this inverse as $\mN$. The matrices $\mL,\mR$ are initialized as all-0 matrices.
	In particular, since we compute $\mM^{-1}$ explicitly, we know $\mM^{-1}_{s,t}$ for all $s,t \in Y$.
	\paragraph{Update to set $Y$.}
	When an index pair $(s,t)$ is added to set $Y$, we must compute $\mM^{-1}_{s,t}$.
	By representation \eqref{eq:inverse} and matrices $\mL,\mR$ consisting of at most $n^\mu$ columns, this takes $O(n^\mu)$ operations.
	When deleting an index pair from $Y$, we do not need to compute anything so such an update takes $O(1)$ time.
	\paragraph{Entry update to $\mM$.}
	When matrix $\mM$ changes, we must update the representation \eqref{eq:inverse}.
	By the Sherman-Morrison identity we have for any two vectors $u,v$ that
	\begin{align}
	(\mM+uv^\top)^{-1} = \mM^{-1} - \frac{\mM^{-1}u~v^\top\mM^{-1}}{1+v^\top\mM^{-1}u}. \label{eq:shermanmorrison}
	\end{align}
	When $\mM$ changes in a single entry, then vectors $u$, $v$ have only one non-zero entry so the products $\mM^{-1}u$ and $v^\top\mM^{-1}$ consist of reading one row and column of $\mM^{-1}$.
	By the representation \eqref{eq:inverse} and $\mR,\mL$ having at most $n^\mu$ columns, this takes $O(n^{1+\mu})$ operations.
	
	Let $u' := \mM^{-1}u / (1+v^\top\mM^{-1}u)$ and $v' = (v^\top\mM^{-1})^\top$, then we can set $\mL \leftarrow [\mL|u']$, $\mR \leftarrow [\mR|v']$ (i.e.~append these vectors to matrices $\mL,\mR$) and thus have
	\begin{align*}
	(\mM+uv^\top)^{-1} 
	=&~
	\mM^{-1} - \frac{\mM^{-1}u~v^\top\mM^{-1}}{1+v^\top\mM^{-1}u} \\
	=&~
	\mM^{-1} - u' v'^\top \\
	=&~
	\mN - \mL \mR^\top - u' v'^\top \\
	=&~
	\mN - [\mL|u'] [\mR|v']^\top
	\end{align*}
	so we again satisfy \eqref{eq:inverse}.
	
	Note that every $n^\mu$ updates, the matrices $\mL,\mR$ have $n^\mu$ columns.
	To guarantee these matrices never have more columns than that, we set $\mN \leftarrow \mN + \mL \mR^\top$ using $O(n^{\omega(1,1,\mu)})$ operations every $n^\mu$ updates, which is $O(n^{\omega(1,1,\mu)-\mu})$ amortized time per update.
	
	At last, note that we must maintain $\mM^{-1}_{s,t}$ for all $s,t \in Y$. This can be done via
	\begin{align*}
		(\mM+uv^\top)^{-1}_{s,t}
		=&~
		\mM^{-1}_{s,t} - u'_s v'_t.
	\end{align*}
	Which takes $O(|Y|)$ time per update.
	
	\paragraph{Low rank update to $\mM$.}
	Consider a rank $n^\delta$ update to $\mM$, i.e.,~we are given two matrices $\mU,\mV$ of size $n\times n^\delta$ and want to update $\mM \leftarrow \mM + \mU \mV^\top$. For these updates we update matrix $\mN$ in \eqref{eq:inverse} but not $\mL$ or $\mR$.
	The Woodbury identity states that
	\begin{align}
		(\mM + \mU \mV^\top)^{-1} = \mM^{-1} - \mM^{-1}\mU (\mI + \mV^\top \mM^{-1} \mU)^{-1} \mV^\top \mM^{-1}. \label{eq:woodbury}
	\end{align}
	Here we can compute $\mM^{-1}\mU = (\mN - \mL \mR^\top) \mU$ and $\mV^\top\mM^{-1} = \mV^\top (\mN - \mL\mR^\top)$ in $O(n^{\omega(1,1,\delta)})$ operations.
	In total, computing $\mM^{-1}\mU (\mI + \mV^\top \mM^{-1} \mU)^{-1} \mV^\top \mM^{-1}$ takes $O(n^{\omega(1,1,\delta)} + n^{\omega(1,\delta,\delta)} + n^{\delta \omega}) = O(n^{\omega(1,1,\delta)})$ operations.
	We subtract the result from $\mN$ in $O(n^2)$ operations, which is subsumed by the previous complexity terms. After this update, we still satisfy \eqref{eq:inverse} because
	\begin{align*}
		(\mM + \mU \mV^\top)^{-1} 
		=&~
		\mM^{-1} - \mM^{-1}\mU (\mI + \mV^\top \mM^{-1} \mU)^{-1} \mV^\top \mM^{-1} \\
		=&~
		\underbrace{\mN - \mM^{-1}\mU (\mI + \mV^\top \mM^{-1} \mU)^{-1} \mV^\top \mM^{-1}}_{\text{new }\mN} - \mL \mR^\top.
	\end{align*}
	\paragraph{Worst-case update time}
	Note that the rank $n^\delta$ update does not affect the size of $\mL$ or $\mR$, so the reset happens after exactly every $n^\mu$ element updates. 
	So the amortized $O(n^{\omega(1,1,\mu)-\mu}+n^{1+\mu}+|Y|)$ update time
	can be made worst-case via standard technique, see e.g.~\cite[Appendix B]{BrandNS19}.
\end{proof}

\Cref{lem:matrix_inverse} can be extended to dynamic matrix inverse for polynomial matrices.
Note that for a polynomial matrix, the inverse might not exist since it's a matrix over a ring instead of a field.
The following \Cref{cor:poly_matrix_inverse} verifies that all the required inverses are well-defined.

\begin{corollary}\label{cor:poly_matrix_inverse}
  Let $\mA$ be a dynamic $n\times n$ matrix receiving entry updates, $h\in[1,n]$, and $\mM := (\mI-X\mA) \in (\F[X]/\langle X^{h+1}\rangle)^{n\times n}$.
	Let $Y\subseteq [n]\times [n]$ be a dynamic set.
	For any $\mu \in [0,1]$, there exists a deterministic data structure maintaining the values $(\mM^{-1})_{s,t}$ for all $(s,t)\in Y$
	and supporting:
	\begin{itemize}
    \item single-entry updates to $\mA$ in $\Ot(h\cdot(n^{\omega(1,\mu,1)-\mu}+n^{1+\mu}+|Y|))$ operations in worst-case,
    \item rank $n^\delta$ updates to $\mA$ for any $\delta \in [0,1]$ in $\Ot(h\cdot n^{\omega(1,\delta,1)})$ operations in worst-case, 
		\item single-element insertions to $Y$ in $\Ot(h\cdot n^\mu)$ operations,
		\item removing elements from $Y$ in $O(1)$ time.
	\end{itemize}
\end{corollary}

\begin{proof}
	The proof follows the common technique of extending dynamic matrix inverse to polynomial matrices \cite{Sankowski04,BrandNS19,BrandFN22}.
	First, note that the inverse of $\mM$ exists and is given by
  \begin{equation*}
    (\mI-X\mA)^{-1} = \sum_{k=0}^{h} X^k \mA^k.
  \end{equation*}
	This is because
  \begin{equation*}
    (\mI - X\mA) \left(\sum_{k=0}^{h} X^k \mA^k\right) = \sum_{k=0}^{h} X^k \mA^k - \sum_{k=1}^{h+1} X^k \mA^k = \mI - X^{h+1} \mA^{h+1} = \mI,
  \end{equation*}
  where the last step uses that the entries of the matrix come from $\F[X]/\langle X^{h+1}\rangle$, i.e.~all arithmetic of the matrix entries is performed modulo $X^{h+1}$, so $X^{h+1}\mA^{h+1}$ is an all-zero matrix.
	
	We obtain \Cref{cor:poly_matrix_inverse} by running \Cref{lem:matrix_inverse} on the matrix $\mM := (\mI - X\mA)$. Note that the complexity increases by a factor of $\Ot(h)$ because the arithmetic operations are now performed on polynomials of degree at most $h$.
	
	The algorithm from \Cref{lem:matrix_inverse} works on polynomial matrices, because the division in \eqref{eq:shermanmorrison} and inversion of \eqref{eq:woodbury} are well defined.
	To see this, note that an update to $\mM = (\mI - X\mA)$ is given by an update to $\mA$.
	So we can assume that vector $u$ in \eqref{eq:shermanmorrison} (and likewise matrix $\mU$ in \eqref{eq:woodbury}) is a multiple of~$X$.
	Thus, the denominator in \eqref{eq:shermanmorrison} is a polynomial of form $(1+X\cdot P(x))$ (and the to be inverted matrix in \eqref{eq:woodbury} is of form $\mI+X\mB$, where $\mB$ is a polynomial matrix) whose inverse is given by
	$\sum_{k=0}^h X^k (-P(X))^k$ (and $\sum_{k=0}^h X^k (-\mB)^k$).
\end{proof}

We can now prove \Cref{t:h-bounded-submatrix-batch} via the reduction from dynamic distance to dynamic matrix inverse by \cite{Sankowski05}.

\begin{proof}[Proof of \Cref{t:h-bounded-submatrix-batch}]
	We maintain the distances in a dynamic graph via a reduction by Sankowski \cite{Sankowski05} (and its deterministic variant from \cite{BrandFN22}).
	Note that 
  \begin{equation*}
    (\mI-X\mA)^{-1} = \sum_{k=0}^h X^k \mA^k \in (\F[X]/\langle X^{h+1} \rangle)^{n\times n}
  \end{equation*}
	and when $\mA$ is an adjacency matrix, we have that $(\mA^k)_{s,t}$ is the number of (not necessarily simple) paths from $s$ to $t$ of length $k$.
	In particular, the smallest $k$ such that $\mA^k_{s,t} \neq 0$ is the $st$-distance.
	So by maintaining the inverse of $(\mI-X\mA)$ for adjacency matrix $\mA$ via data structure \Cref{lem:matrix_inverse} we can maintain $h$-bounded distances by just checking the entries of $(\mI-X\mA)^{-1}$.
	
	The maximum number of paths is bounded by $n^h$, so we can pick as field $\F = \Z_p$ for some prime of size $n^h \le p \le 2n^h$. Here every arithmetic operation in $\Z_p$ takes $O(h)$ time in the Word-RAM model with word size $O(\log n)$.
	Thus the time complexity increases by a factor $O(h)$ compared to \Cref{cor:poly_matrix_inverse}.
\end{proof}

\section{Omitted proofs}\label{s:omitted}

\lblocks*
\begin{proof}
  Suppose we identify $V$ with $\{1,\ldots,n\}$. Let us partition $V$ into $q=O(n/b)$ blocks $V_1,\ldots,V_q$ of~$\leq b$ consecutive
  integers. Moreover, let $V'$ be a copy of $V$, $V\cap V'=\emptyset$, such that each $v\in V$ has a unique corresponding vertex $v'\in V'$.
  Set $V_{\leq i}:=\bigcup_{j=1}^i V_j$ for $i\geq 0$.
  For each $i=1,\ldots,q$, define:
  \begin{equation*}
    G_i:=\left(V\cup V',E(G)\cup \left\{uv':uv\in E(G)\text{ and }u\in V_{\leq i}\right\}\right).
  \end{equation*}
  Note that with graphs $G_i$ defined as above, a single-edge update to $G$ leads to at most two single-edge
  updates in each $G_i$. Similarly, if vertex updates are supported, then a vertex update to $G$
  can be reflected in each $G_i$ using $O(1)$ vertex updates therein.
  Moreover, we have $G_q=G$.
  
  Observe that for any $u,v\in V$, $u\neq v$, there is a 1-1 correspondence between $u\to v$ paths $P=P'\cdot xv$ in $G$
  with $x\in V_{\leq i}$ and paths $u\to v'=P'\cdot xv'$ in $G_i$.
  As a result, we have $\dist_G^h(u,v)=\dist_{G_i}^h(u,v')$ iff some shortest $h$-bounded path in
  $G$ has the penultimate vertex in $V_{\leq i}$.
  Moreover, every path in $G$ is preserved in $G_i[V]$, so $\dist_G^h(u,v)=\dist_{G_i}^h(u,v)$ for all $i$.

  We construct and maintain data structures $\mathcal{D}(G_1),\ldots,\mathcal{D}(G_q)$.
  Since the $q=O(n/b)$ graphs $G_i$ are at most twice as large as $G$, the initialization
  time is $O\left(\frac{n}{b}\cdot \isource(n,h)\right)$ and the update time is
  $O\left(\frac{n}{b}\cdot \usource(n,h)\right)$.

  Upon query $(S,Y)$, for each $s\in S$ we issue a single-source query to each $\mathcal{D}(G_i)$,
  so that the values $\dist_{G_i}^h(s,t)$ for all $(s,t)\in S\times V$ and $i$ are computed.
  This takes $O((n/b)\cdot |S|\cdot \qsource(n,h))$ time.
  Within the same bound, for each $(s,t)\in S\times V$, we (naively) locate the smallest index $i_{s,t}$ such that $\dist_{G_{i_{s,t}}}^h(s,t')=\dist_{G_q}^h(s,t')$.

  For each $(s,t)\in Y$, we find the predecessor on some shortest $h$-bounded $s\to t$ path in $G$ as follows.
  Find the first vertex $p_{s,t}\in V_{i_{s,t}}$ such that $p_{s,t}t\in E(G)$ and $\dist_{G_q}^h(s,t')=\dist_{G_{q}}^h(s,p_{s,t})+1$, and report
  $p_{s,t}$ as the sought predecessor. Since $|V_{i_{s,t}}|\leq b$, through all $(s,t)$ this takes $O(|Y|b)$ time.
  
  We now argue that the above is correct. By the definition
  of $i_{s,t}$, we have either $i_{s,t}=1$ or $\dist_{G_{i_{s,t}-1}}^h(s,t)>\dist_{G_q}^h(s,t)$.
  As a result, (1) no shortest $h$-bounded $s\to t$ path in $G=G_q$ has the penultimate vertex in $V_{\leq (i_{s,t}-1)}$,
  (2) at least one such path has the penultimate vertex in $V_{\leq i_{s,t}}$.
  But $V_{\leq i_{s,t}}\setminus  V_{\leq (i_{s,t}-1)}=V_{i_{s,t}}$ and thus the penultimate vertex $z$ on each such path
  has to satisfy $z\in V_{i_{s,t}}$, $zt\in E(G)$ and $\dist_{G}^h(s,z)+1=\dist_G^h(s,t)$.
  On the other hand, any $z$ satisfying this condition is a valid predecessor vertex.
  By $\dist_{G_q}^h(s,t')=\dist_G^h(s,t)$ and $\dist_{G_q}^h(s,z)=\dist_G^h(s,z)$ we obtain
  that the found vertex $p_{s,t}$ is a valid predecessor.
\end{proof}

\lblockstwo*
\begin{proof}
  We will only describe the changes needed to make to the construction
  of Lemma~\ref{l:blocks}. We use only one data structure $\mathcal{D}(G)$;
  for the $q=O(n/b)$ auxiliary graphs $G_i$, we use data structures
  $\mathcal{D}'(G_1),\ldots,\mathcal{D}'(G_q)$.
  The update and initialization bounds follow easily.

  Recall that in the proof of Lemma~\ref{l:blocks}, once we have computed,
  for each $(s,t)\in Y$, the smallest index $i_{s,t}$ such that $\dist^h_{G_{i_{s,t}}}(s,t')=\dist_{G_q}^h(s,t')=\dist_G^h(s,t)$, locating the desired predecessors
  required $O(|Y|b)$ additional time.
  The values $\dist_G^h(s,t)$ for all $(s,t)\in Y$ can be clearly computed
  by issuing $|S|$ single-source queries to $\mathcal{D}(G)$, i.e., in
  $O(|S|\cdot \qsource(n,h))$ time.
  By $G_1\subseteq G_2\subseteq \ldots\subseteq G_q$,
  we have $\dist^h_{G_1}(s,t')\geq \dist^h_{G_2}(s,t')\geq \ldots \geq \dist^h_{G_q}(s,t')=\dist^h_{G}(s,t)$.
  As a result, the desired smallest $i_{s,t}$ such that $\dist_{G_{i_{s,t}}}^h(s,t')=\dist_G^h(s,t)$ can be located via binary search using $O(\log{n})$
  queries about $\dist_{G_i}^h(s,t')$ issued to $\mathcal{D}'({G_i})$.
  Each of these queries costs $\qpair(n,h)$ time.
\end{proof}

\linducedsubgraph*
\begin{proof}
  Let a graph $G_W=(V',E')$ be obtained as follows. For each $v\in V$, there
  are two copies $v^+$ and $v^-$ of $v$ in $V'$, and an auxiliary edge $e_v=v^+v^-\in E'$ iff $v\in W$.
  Moreover, each edge $uv\in E$ gives rise to the edge $u^-v^+$ in $E'$.
  This way, $|V'|=2|V|$ and $|E'|=|E|+|W|$.

  Clearly, an edge update to $G$ can be simulated using a single edge update to $G_W$.
  Similarly, a vertex update to $G$ can be simulated using $O(1)$ vertex updates to $G_W$.
  Finally, an element update to $W$ translates to a single auxiliary edge update to $G_W$.
  We set $\mathcal{D}'(G):=\mathcal{D}(G_W)$.

  Observe a 1-1 correspondence of $k$-edge paths $P=s\to t$ in $G[W]$ and $2k$-edge
  paths $s^+\to t^+$ in $G_W$.
  Moreover, no odd-length $s^+\to t^+$ path exists in $G_W$.
  As a result, we have $\dist_G^h(s,t)=\frac{1}{2}\cdot \dist_G^{2h}(s^+,t^+)$.
  The submatrix predecessor query $(S,Y)$ in $G[W]$ is translated
  to the query $\left(\{s^+:s\in S\},\{(s^+,t^+):(s,t)\in Y\}\right)$
  issued to $G_W$.
  By the path correspondence and the fact that all incoming edges of $t^+$ in $G_W$ correspond
  to original incoming edges of $t$ in $G$, we can indeed answer a submatrix predecessor
  query on $G[W]$ in $O(Q(|S|,|Y|))$ time.
\end{proof}

\end{document}